\documentclass[11pt]{amsart}

\usepackage[showonlyrefs]{mathtools}
\usepackage{mathdots}
\mathtoolsset{showonlyrefs=true}
\usepackage{hyperref}
\usepackage{xcolor}
\usepackage{bbm}
\usepackage{subcaption}
\theoremstyle{plain}
\newtheorem{theorem}{Theorem}[section]

\newtheorem{corollary}[theorem]{Corollary}
\newtheorem{example}[theorem]{Example}
\newtheorem{lemma}[theorem]{Lemma}
\newtheorem{definition}[theorem]{Definition}
\newtheorem{proposition}[theorem]{Proposition}
\theoremstyle{remark}
\newtheorem{remark}[theorem]{Remark}

\renewcommand{\Re}{\mathop{\mathrm{Re}}}
\renewcommand{\Im}{\mathop{\mathrm{Im}}}

\DeclareMathOperator{\Tr}{Tr}

\numberwithin{equation}{section}

\author{Nedialko Bradinoff \and Maurice Duits}
\address{Royal Institute of Technology (KTH), Lindstedtsvägen 25, 10044, Stockholm, Sweden}
\email{nedialko@kth.se}
\title[Toeplitz determinants with slow Fourier decay]{On Toeplitz determinants and CUE linear statistics with slow Fourier decay}
\address{Royal Institute of Technology (KTH), Lindstedtsvägen 25, 10044, Stockholm, Sweden\\
\emph{Current address:} Sapienza Università di Roma
Dipartimento di Matematica Guido Castelnuovo,
Piazzale Aldo Moro n. 5,
00185 Roma, Italy}
\email{duits@kth.se,maurice.duits@uniroma1.it}
\date{}
\begin{document}
\vspace*{-11pt}
\begin{abstract}
We study Toeplitz determinants $\det T_n(e^{tf})$ with $t \in \mathbb C$ and functions $f$ whose Fourier coefficients satisfy $f_k=\mathcal O(|k|^{-1})$. This regime extends beyond $H^{1/2}$ and includes symbols with Fisher--Hartwig singularities. Our approach is operator-theoretic and founded on the Baker--Campbell--Hausdorff formula that separates the quadratic term $ \sum_{k=1}^{\infty}\min(k,n)f_kf_{-k} $ from the higher-order terms in the expansion of $\log\det T_n(e^{tf})$. We show that this quadratic term accounts for the possible growth with $n$, while every fixed higher-order coefficient remains bounded. For symbols with bounded positive and negative Fourier parts, our estimates yield two-sided bounds for the determinant after removal of the quadratic contribution. For a broader admissible class, including Fisher--Hartwig-type symbols, we obtain uniform higher-order coefficient bounds and a central limit theorem for the associated CUE linear statistics.  We also obtain bounds on mixed exponential moments for CUE-derived random fields beyond the characteristic polynomial.
\end{abstract}
\maketitle
\tableofcontents
\section{Introduction}
Let $\mathbb T= \{z \in \mathbb C \mid |z|=1\}$ be the complex unit circle. For $a \in \mathbb L_1(\mathbb T)$, with Fourier coefficients given by 
\begin{equation} \label{eq:def_fourier_coefficient}
    a_k=\frac{1}{2 \pi } \int_0^{2 \pi} a(e^{i \theta}) e^{-ik\theta}  \ d \theta, \qquad k \in \mathbb Z,
\end{equation}
the Toeplitz matrix $T_n(a)$ of size $n \in \mathbb N$ is defined as 
\begin{equation}\label{eq:def_Toeplitz_matrix}
    T_n(a)= (a_{j-k})_{j,k=1}^n.
\end{equation}
The function $a$ is called the \textit{symbol} of the Toeplitz matrix. 

In this paper, motivated by applications to random matrix theory, we will be interested in the asymptotic behavior, as $n\to \infty$, of the determinant 
\begin{equation}\label{eq:def_Toeplitz_determinant}
 D_n(e^{tf}):= \det T_n(e^{t f}),
\end{equation}
for $t \in \mathbb C$ and $f$ is a function such that   
\begin{equation}\label{eqn:bound_on-log}
\limsup_{k\to \pm \infty}  |k||f_k|< \infty.
\end{equation}
We will be specifically interested in the case of slowly decaying Fourier coefficients where $|f_k|\sim 1/|k|$ as $k \to \pm \infty$. Our main results are on the coefficients $ C_m^{(n)}(f)$ in the expansion around $t=0$, 
$$
\log D_n(e^{tf})= \sum_{m=1}^\infty t^m C_m^{(n)}(f).
$$
In particular, we prove bounds on  $C_m^{(n)}(f)$ (uniform in $n$ and $f$).

The  motivation for the condition \eqref{eqn:bound_on-log} is that it concerns symbols that lie just beyond the classical scope of the Strong Szeg\H{o} Limit Theorem, including Fisher--Hartwig singularities as a special case.  We will recall these results in more detail below, but already here mention that 
the Strong Szeg\H{o} Limit Theorem shows that  $C_2^{(n)}(f)$ converges to a non-zero limit, as $n \to \infty$, and $C_m^{(n)}(f) \to 0$ for $m\geq 3$.   For symbols with Fisher-Hartwig singularities, the $m=2$ term grows logarithmically and the terms with  $m\geq 3$ no longer converge to zero but do remain bounded.  

This has an important probabilistic interpretation  as  $\log D_n(e^{tf})$ is the cumulant generating function for the linear statistic associated with the test function $f$ for random unitary matrices (cf \eqref{eq:linear_statistic_toeplitz_determinant}). 
For smooth test functions, the Strong Szeg\H{o} Limit Theorem implies that these statistics satisfy a CLT: cumulants of order higher than two vanish asymptotically, and the variance converges to the square of the $H^{\frac12}$-norm. For test functions with Fisher--Hartwig singularities, these cumulants do not vanish, but remain bounded, while the variance grows logarithmically. Interestingly, this phenomenon  also appears in different models from statistical physics, as discussed by Conlon-Spencer \cite{CS}. \emph{We show that the mechanism behind the phenomenon that the variance grows but higher cumulants remain bounded,  is a consequence of the slow decay of the Fourier coefficients (the presence of high frequencies). Fisher–Hartwig singularities are a natural source of such decay, but the class we consider also includes continuous functions, so singularities themselves are not what drives this phenomenon.}

Before stating our main results, we will discuss these two motivations and their relevance in more detail.

\subsection{The Strong Szeg\H{o} Limit Theorem and Fisher-Hartwig asymptotics}
We begin by recalling the Strong Szeg\H{o} Limit Theorem. For sufficiently smooth (and possibly complex-valued) functions $f$, we have  (dropping $t$ from the expressions for the moment)
\begin{equation}\label{eq:SSLT}
    \det T_n(e^{f})= e^{n f_0+ \sum_{k=1}^\infty k f_k f_{-k}}(1+o(1)),
\end{equation}
as $n \to \infty$. The Strong Szeg\H{o} Limit Theorem has a long history, and  we refer to \cite{DIK} for an extensive discussion. The first result was obtained by Szeg\H{o} \cite{Sz3} for positive  $a=e^f \in C^{1+\varepsilon}$ for any $\varepsilon>0$. The assumptions on the symbol were relaxed by various authors in subsequent works. The strongest version of the Strong Szeg\H{o} Limit Theorem is due to Johansson \cite{Joh1} who only required $f$ to satisfy
\begin{equation}\label{eq:boundedHhalfnorm}
\sum_{k=-\infty}^\infty |k| |f_k|^2 < \infty,
\end{equation}
improving the result by Widom \cite{Wid4} who worked under the additional assumption that $f$ is essentially bounded. 

Clearly, when $f_k\sim 1/k$ as $k \to \infty$, \eqref{eq:boundedHhalfnorm} fails.  A special case is that of the Fisher-Hartwig class of symbols, with singularities at distinct locations $z_1, \ldots, z_m\in\mathbb{T}$ for $m\geq 1$,
\begin{equation}\label{eqn:fh_log}
f(z)
=V(z)+\sum_{j=1}^{m}\left( \delta_j\log\left(1-{z_j}z^{-1}\right)+\gamma_j\log\left(1-z_j^{-1}z\right) \right),
\end{equation}
where
$V(z)$ is a smooth function and
for $j=1,\ldots, m$, $\delta_j,\gamma_j\in\mathbb{C}$ with
$\Re(\delta_j+\gamma_j)> -1$ and 
the assumed branches of logarithm are chosen so that for $z\in\mathbb{T}$,
$$\log(1-z_jz^{-1})=-\sum_{k=1}^{\infty} \frac{z_j^kz^{-k}}{k}\quad\text{and}\quad 
\log(1-z_j^{-1}z)=-\sum_{k=1}^{\infty} \frac{z^kz_j^{-k}}{k}.
$$

In that case, it is known that, setting $\delta_j+\gamma_j=2\alpha_j$, and $\gamma_j-\delta_j=2\beta_j$, 
for $\Re(\alpha_j)>-\frac{1}{2}$, 
$\alpha_j\pm\beta_j\not= -1, -2,\ldots$  and $\max_{j,k}|\Re(\beta_j)-\Re(\beta_k)|<1$, then
\begin{multline}\label{eq:FH}
\det T_n(e^{f})
\\=
e^{n f_0}
n^{\sum_{j=1}^{m}
(\alpha_j^2-\beta_j^2)}
\prod_{1\leq j<k\leq m}
\left| z_j-z_k \right|^{2(\beta_j\beta_k-\alpha_j\alpha_k)}\left(\frac{z_k}{z_j e^{i\pi}}\right)^{\alpha_j\beta_k-\alpha_k\beta_j}
\\\times
G(e^f)(1+o(1)),
\end{multline}
as $n \to \infty$. Here the constant $G(e^f)$ is a complicated expression involving Barnes G-functions. We will omit it here and refer to \cite{DIK} for the exact form. 

Fisher--Hartwig asymptotics play a central role in several areas of mathematics and mathematical physics, and we refer to \cite{DIK} for a detailed historical survey. Here we only recall a few milestones. Motivated by Lenard's study of the ground state of impenetrable bosons \cite{Len1}, Fisher and Hartwig formulated their celebrated conjecture describing the asymptotics of Toeplitz determinants with singular symbols \cite{FisHart1}. 
It took a considerable amount of work to prove the result in its full generality,
\cite{Bas1,Bas3,BasHelt,BottSilb,Wid2}, 
culminating in the complete results obtained in \cite{DIK1,Ehr}.
Another classical application arises in the two-dimensional Ising model, where Toeplitz determinants with Fisher--Hartwig singularities appear in the analysis of spin--spin correlation functions, \cite{KauOns}. More recently, Fisher--Hartwig asymptotics have become fundamental in random matrix theory and the  study of the  characteristic polynomial of the Circular Unitary Ensemble (CUE). This application was an important motivation for the present paper and will be discussed in more detail below. But before we come to that, let us first state our main results informally.

\subsection{Informal discussion of main results}\label{sec:informal}

The key structural difference between the Strong Szeg\H{o} Limit Theorem
\eqref{eq:SSLT} and the Fisher--Hartwig asymptotics \eqref{eq:FH} is the
appearance of an additional power of $n$ on the right-hand side of
\eqref{eq:FH}. This raises the question of  whether this  
additional power of $n$ is a universal phenomenon or merely a special feature
of Fisher--Hartwig singularities. Moreover, inspection of \eqref{eq:FH} shows that its right-hand side becomes
singular as two Fisher--Hartwig singularities approach one another, making this
representation ill-suited to the study of merging singularities. Such merging regimes are important in applications to characteristic
polynomials, and interesting scaling limits involving Painlev\'e transcendents have been found by Claeys-Krasovsky \cite{ClaKra}. However, the appearance of Painlev\'e transcendents, although interesting in its own right, is often not essential for these applications. This motivated Fahs \cite{Fahs} to adjust the Riemann-Hilbert analysis to obtain more uniform asymptotics, without explicit constant terms, simultaneously settling a conjecture by Fyodorov-Keating. However, this approach is still tailored to local analysis for Fisher-Hartwig type of singularities and does not easily extend beyond this class.

 In this paper, we propose an operator-theoretic approach, based on the Baker--Campbell--Hausdorff formula, which separates the quadratic contribution from the higher-order terms and controls the latter directly in terms of the decay of the Fourier coefficients. This leads to estimates that are uniform in the locations of the singularities and remain applicable beyond the Fisher--Hartwig class. The gain in robustness and generality comes at the expense of the explicit constant term and the full asymptotic expansion.

The starting point of our analysis is the representation:
\begin{equation}\label{eqn:nvariance}
\det T_n(e^{tf})
=
\exp\left(
nt f_0+t^2\sum_{k=1}^{\infty}\min(k,n)f_kf_{-k}
\right)
F_n(e^{tf}),
\end{equation}
and we then derive quantitative bounds on the remainder $F_n$. Note that in case $f\in H^{\frac12}$, the Strong Szeg\H{o} Limit Theorem
\eqref{eq:SSLT} implies that
$
F_n(e^{tf})\to 1,$ as 
$ n\to\infty.$ 

The appearance of $\min(k,n)$, instead of simply $k$, in the quadratic term in the exponent at the right-hand side of \eqref{eqn:nvariance} was crucial for   Duits-Johansson \cite{DuiJoh} who showed that $F_n(e^{tf_n})\rightarrow 1$ even for a class of symbols that vary with $n$ in a particular way.
In our context, it is important to note that it also effectively deals with merging Fisher-Hartwig  singularities as it encodes both  the term $n^{\sum_{j=1}^{m}
(\alpha_j^2-\beta_j^2)}$ and the term 
 $$\prod_{1\leq j<k\leq m}
\left| z_j-z_k \right|^{2(\beta_j\beta_k-\alpha_j\alpha_k)}\left(\frac{z_k}{z_j e^{i\pi}}\right)^{\alpha_j\beta_k-\alpha_k\beta_j}.$$
 Indeed, we invite the reader to check (see the proof of Corollary \ref{cor:mixed_exponential_moments} below for an explicit computation) that for symbols in the Fisher-Hartwig class we have 
\begin{equation}\label{eqn:fourier_dec}
f_k= -\frac{1}{|k|}\sum_{j=1}^m \kappa_{j,\pm} e^{-ik\theta_j}+o\left(\frac{1}{|k|}\right),
\end{equation}
with $\kappa_{j,-}=\delta_j$ for $k<0$, and $\kappa_{j,+}=\gamma_j$ for $k>0$,
and  that in this concrete setting 
\eqref{eqn:nvariance} reproduces both terms up to a uniformly bounded remainder.  That is, there exist constants $A_1,A_2>0$ such that
$$A_1<|F_n(e^{tf})|<A_2$$
for $n \in \mathbb N$ and $t$ in a neighborhood of the origin.   An important feature of our point of view is that we have bounds that are uniform in the location of the singularities.

Very roughly speaking, our main results show that under the general  condition \eqref{eqn:bound_on-log}, the remainder $F_n(e^{tf})$ admits bounds of the type suggested by the Fisher--Hartwig case, although the estimates obtained here do not in general yield a full asymptotic formula. More precisely, in the expansion
$$
\log F_n(e^{tf})= \sum_{m=3}^\infty t^m C_m^{(n)}(f),
$$
we prove bounds on the coefficients $C_m^{(n)}(f)$.

Before stating our main results precisely, we first discuss the connection with
the spectra of random matrices, which further illustrates the motivation for
the present work and gives a probabilistic motivation for the representation in \eqref{eqn:nvariance}.

\subsection{The CUE and linear statistics}

For $n \in \mathbb N$, the Circular Unitary Ensemble (CUE) of size $n$ is the probability measure on the group of $n\times n$ unitary matrices equipped with the normalized Haar measure. 
For a randomly chosen unitary matrix, we denote the eigenvalues by $e^{i \theta_1},\ldots,e^{i \theta_n}$. By a classical result of Weyl, the jpdf of the eigenvalues is given by 
\begin{equation}\label{eq:jpdf_eigenvalues_CUE}
    \frac{1}{(2 \pi)^n n! } \prod_{1\leq j< k \leq n }|e^{i \theta_j}-e^{i \theta_k}|^2 d \theta_1 \cdots d \theta_n.
\end{equation}
This jpdf implies that we can think of the eigenvalues as a log-gas with $n$ particles without an external field at inverse temperature $\beta=2$. 
In particular, the eigenvalues repel each other. In fact, this repulsion is very effective, as can be witnessed from the CLT for linear statistics. 
For a function $f$, the linear statistic associated to $f$ is the random variable $X_n(f)$ defined by 
\begin{equation}\label{eq:CLT_linear_statistic}
    X_n(f)=\sum_{j=1}^n f(e^{i \theta_j}).
\end{equation}
If $f$ is sufficiently smooth (and real-valued), then, as $n \to \infty$,
\begin{equation} \label{eq:CLT_CUE}
    X_n(f)-\mathbb E X_n(f) \to Z \sim N\left(0,2 \sum_{k=1}^\infty k |f_k|^2\right),
\end{equation}
in distribution.
The remarkable feature of this CLT is the fact that there is no normalization of the linear statistic. 

The CLT, \eqref{eq:CLT_CUE}, was  proved in the case of trigonometric polynomials in the work by Diaconis-Shahshahani \cite{DiaconisShah} and extended by Diaconis-Evans  \cite{DiaconisEvans}, using connections with representation theory. 
Prior to that,  Johansson \cite{Joh1} had in fact already observed that the moment generating function of $X_n(f)$ is, by Andr\'eief's identity (i.e. the general version of the Cauchy-Binet identity), given by the Toeplitz determinant:
\begin{equation}\label{eq:linear_statistic_toeplitz_determinant}
    \mathbb E[e^{t X_n(f)}]=\det T_n(e^{tf}).
\end{equation}
The CLT \eqref{eq:CLT_CUE} is then immediately seen to be equivalent to the Strong Szeg\H{o} Limit Theorem  for Toeplitz determinants. Later inspired by \cite{DiaconisShah}, Johansson exploited this connection further to prove a conjecture by Diaconis that the rate of convergence is super-exponential, \cite{Joh2}.  We refer to \cite{Diaconis} for a beautiful survey and discussion on this topic. 

Fisher-Hartwig asymptotics also have an interpretation along these lines, which is easiest to observe when we take logarithms,
$$
\log \det T_n(e^{tf})=\log \mathbb E[e^{t X_n(f)}]=\sum_{m=1}^\infty { C_m^{(n)}(X_n(f)) t^m},$$
where $C_m^{(n)}(X_n(f))$ are, up to a factor of $m!$, the cumulants of the linear statistic $X_n(f)$.\footnote{In the literature, the cumulants $\kappa_m$ are usually defined by
$
\kappa_m = m! C_m^{(n)}.
$
We prefer to work with $C_m^{(n)}$, as this leads to cleaner statements in our setting.}
The Fisher--Hartwig asymptotics \eqref{eq:FH} imply that if $f$ is smooth apart from logarithmic singularities or jump discontinuities, then the variance grows logarithmically with $n$, while the cumulants of order three and higher remain uniformly bounded in $n$ (but do not tend to $0$). This phenomenon was also discussed by Conlon-Spencer in \cite{CS}, where they proved a similar result for height differences in a class of two-dimensional uniformly convex $\nabla\phi$ models and discussed its universality. We add to this discussion by extending the analysis beyond the Fisher–Hartwig class, to symbols governed more generally by the decay of their Fourier coefficients.

\subsection{The Characteristic Polynomial, log-correlated fields, and Multiplicative Chaos}\label{subsec:logcorfield}

An important interpretation of the CLT in \eqref{eq:CLT_CUE} is that the logarithm of the characteristic polynomial of a CUE matrix converges to a Gaussian log-correlated field (as was first discussed in \cite{HKOC}). For  a Haar distributed $n\times n$ unitary matrix $U$, set 
\begin{equation}\label{eq:def_char_pol}
    \Phi_n(e^{i \theta}) = \log |\det (U-e^{i \theta})|, \qquad \theta \in [0,2 \pi).
\end{equation}
(we will only consider the real part of the logarithm of the characteristic polynomial, but an analogous story holds for the imaginary part). The interpretation of \eqref{eq:CLT_CUE} is now that the field $(\Phi_n(e^{i \theta}))$ converges, as $n \to \infty$, to the Gaussian log-correlated field $(\Phi(e^{i \theta}))_{\theta\in [0,2 \pi)}$ with covariance  
\begin{equation}
    \label{eq:covariance_structure}
    C(\theta,\theta')
    =\mathbb E \left[
    \Phi(e^{i \theta})\Phi(e^{i \theta'})
    \right]=-\frac{1}{2}\log|e^{i \theta}-e^{i \theta'}|. 
\end{equation}
Care should be taken here, since log-correlated fields cannot be realized as random functions, but have to be interpreted in a distributional sense. 
The standard way of doing this is by pairing $\Phi_n$ with a smooth test function $\phi$:
\begin{equation}
    \label{eq:pairing}
    \langle \Phi_n, \phi\rangle= \int \Phi_n(e^{i \theta}) \phi(e^{i \theta})d \theta
    = \sum_{j=1}^n \int \phi(e^{i \theta}) \log |e^{i \theta}-e^{i \theta_j}|\  d \theta,
\end{equation}
and then showing that, as $n\to \infty$,
\begin{equation}\label{eq:gaussian_limit_field}
 \langle \Phi_n, \phi\rangle-\mathbb E \left[ \langle \Phi_n, \phi\rangle\right]
 \to N\left(0,\iint \phi(e^{i \theta}) \phi(e^{i \theta'}) C(\theta,\theta') d\theta d\theta' \right)
\end{equation}
in distribution, with $ C(\theta,\theta')$ as given in \eqref{eq:covariance_structure}. Note that, due to the second identity in \eqref{eq:pairing},
the pairing $\langle \Phi_n, \phi\rangle$ equals the linear statistic $X_n(f)$ with  $$f(e^{i \theta'})=  \int \phi(e^{i \theta}) \log |e^{i \theta}-e^{i \theta'}|\  d \theta.$$
Since $\phi$ is a smooth test function, $f$ is also smooth, and with some additional arguments, the CLT, \eqref{eq:CLT_CUE}, can be shown to imply \eqref{eq:gaussian_limit_field}. 

The connection with log-correlated fields can be taken one step further.
The convergence of $\Phi_n$ to a
Gaussian log-correlated field suggests the stronger statement that the random
measures
\begin{equation}\label{eq:gmc_measure_char_poly}
\mu_{n,\gamma}(d\theta)
=
\exp\left\{
\gamma \Phi_n(\theta)
-
\frac{\gamma^2}{2}\mathbb E\left[\Phi_n(\theta)^2\right]
\right\}\,d\theta
\end{equation}
should converge, as $n\to\infty$, to a Gaussian multiplicative chaos measure,
at least in the subcritical range $0<\gamma<2$. Gaussian multiplicative chaos was introduced by Kahane \cite{Kahane} and further developed by Rhodes-Vargas,  see \cite{RV} for a modern comprehensive review.  For the CUE this convergence was proved in the $L^2$-phase by Webb \cite{W}, and later extended to the full subcritical regime by Nikula, Saksman and Webb \cite{NSW}. Related convergence results have also been obtained for other random matrix ensembles, including unitary invariant Hermitian ensembles \cite{BWW}.

Besides providing a natural refinement of the Gaussian field limit, this
perspective is closely connected to the study of the Riemann zeta function on
the critical line. The analogy between the zeros of $\zeta$ and the eigenvalues
of random unitary matrices goes back to Montgomery's pair-correlation
conjecture \cite{M}. The comparison between values of $\zeta(1/2+it)$ and
characteristic polynomials of CUE matrices was then developed systematically by
Keating and Snaith \cite{KS}. In later work, Fyodorov, Hiary and Keating used this connection to formulate
precise predictions for the extreme values of $\zeta$ on short intervals, based
on the freezing transition for log-correlated random energy models and Gaussian
multiplicative chaos \cite{FHK,FK}.

It is important, however, to note that Gaussian multiplicative chaos is a
delicate object, and that the convergence in \eqref{eq:gaussian_limit_field} is
not by itself strong enough to imply convergence of the measures
$\mu_{n,\gamma}$. A key difficulty is that, for finite $n$, the field $\Phi_n$
is not Gaussian. Thus $\Phi_n$ cannot simply be viewed as a Gaussian
mollification of the limiting log-correlated field (which would allow for the direct approach in \cite{Ber}). Indeed, one needs control of mixed
exponential moments of the form
\begin{equation}\label{eq:mixed_exponential_moment}
\mathbb E\exp\left\{
\sum_{j=1}^m \gamma_j \Phi_n(\theta_j)
-
\frac12
\sum_{j=1}^m \gamma_j^2
\mathbb E\left[\Phi_n(\theta_j)^2\right]
\right\},
\end{equation}
particularly for $m=2$. Moreover, the control needs to be uniform in both $n$ and the locations of the singularities at $\theta_j$. By \eqref{eq:linear_statistic_toeplitz_determinant} this is a ratio of Toeplitz determinants with Fisher-Hartwig type of singularities. However, the Fisher-Hartwig asymptotics do not cover the regime where singularities are at distance $n^{-\delta}$ with $0<\delta<1$, which is needed in the analysis (see for example \cite{CFLW}).

Motivated by this application, we seek an approach that sacrifices asymptotic precision in exchange for estimates that remain uniform in the relevant parameters. Our methods also extend beyond characteristic polynomials to other random fields exhibiting closely related behavior; see Section \ref{sec:bounded_dec}, preceding Corollary \ref{cor:mixed_exponential_moments}.
\subsection{Acknowledgements}
The authors were supported by the European Research Council (ERC), Grant Agreement No. 101002013.
Part of this work was completed while MD held a Chaire d’Excellence from the Fondation Sciences Math\'ematiques
de Paris (FSMP). The authors thank the LPSM at Sorbonne University for its hospitality during this period.
\subsection{AI disclosure}
No AI was used in developing the mathematical arguments in this article; all arguments are entirely the authors' own.  We do not expect our results to be optimal: the regularity assumptions can likely be weakened and stronger results obtained. After posting the first version of this work on the arXiv in August 2026, we consulted with an AI system and indeed it appears  possible to push our strategy further and obtain stronger results. We have nevertheless chosen not to incorporate AI-generated arguments, and to keep the paper as we originally conceived it. We are proud of our hard labor, and prefer to separate it from  AI-assisted developments that we plan to pursue in future work.
\section{Statement of main results}
In this section we state our main results; the proofs are postponed to later sections. Let $f$ be a function on $\mathbb T$ whose Fourier coefficients satisfy
\begin{equation}\label{eq:condition_on_fk}
|f_k| = \mathcal O(1/|k|), \qquad k \to \pm \infty.
\end{equation}
For $t$ in a sufficiently small neighborhood of the origin, we consider
\begin{equation}\label{eqn:cumulant_generating_function}
\Psi_n(t)=\log \det T_n(e^{tf}).
\end{equation}
Under the assumptions above, $\Psi_n$ is analytic near $t=0$, and we write
$$
\Psi_n(t)=\sum_{m=1}^{\infty}C_m^{(n)}(f)t^m.
$$
Our main results identify the quadratic coefficient $C_2^{(n)}(f)$ and give bounds on $C_m^{(n)}(f)$ for $m\ge3$, uniform in $n$.

For the rest of the paper, we will assume that $f_0=0$, as this Fourier coefficient contributes only a trivial factor to the Toeplitz determinant.

\subsection{Symbols with bounded decomposition} \label{sec:bounded_dec}

For a $L^2(\mathbb T)$ function $f$ with Fourier coefficients $f_k$, we define 
$$
f_+(z)= \sum_{k=1}^\infty f_k z^k, \qquad  f_-(z)=\sum_{k=1}^\infty f_{-k}z^{-k}.
$$
Then our strongest bounds hold for symbols for which both $f_\pm$ are (essentially) bounded. Note that under this assumption $\Psi_n(t)$ is well defined and analytic as a function of $t$ in a neighborhood of the origin.

\begin{theorem}\label{thm:cumulants_bounded_and_decay}
Suppose that $f_{\pm}$ are essentially bounded by some $C>0$. Assume further that 
for some $\kappa>0$, $|f_k|<\kappa{|k|}^{-1}$. 
Then, for $n\in \mathbb N$,
\begin{equation}
    C_2^{(n)}(f)= \sum_{k=1}^\infty \min (k,n) f_kf_{-k}.  
\end{equation}
and, with $m\geq 3$,
\begin{equation}
|C_m^{(n)}(f)|\leq 
a c^m,
\end{equation}
for some constants $a,c>0$ that depend  on $ \kappa$ and $C$ but not on $m$ or $n$.
\end{theorem}
\begin{remark}
   In Corollary~\ref{cor:bounds_in_bounded} below, we provide explicit values for these constants. Since we do not expect those values to be optimal, we have deliberately left them unspecified here.
\end{remark}

The proof of Theorem~\ref{thm:cumulants_bounded_and_decay} is given in Section~\ref{sec:proofs}. The boundedness of $f_\pm$ is convenient at several points in the argument, but it excludes Fisher--Hartwig symbols, for which $f_\pm$ necessarily have logarithmic singularities. Nevertheless, there are broad classes of functions satisfying \eqref{eq:condition_on_fk} for which $f_\pm$ are bounded; examples are presented in Section~\ref{subsec:examples}.

\begin{corollary}\label{cor:determinant_expansion}
Suppose $f:\mathbb{T}\rightarrow\mathbb{C}$ satisfying \eqref{eq:condition_on_fk},  and  $\|f_+\|_{\infty}, \|f_-\|_{\infty}< \infty$. Then $\det T_n(e^{tf})$ is analytic in $t$  in a small neighborhood of $t=0$, and has the representation
\begin{equation}
\det T_n(e^{tf})=\exp\left(tnf_0+ t^2\sum_{k=1}^\infty \min (k,n) f_kf_{-k}\right) F_n(t),
\end{equation}
where $F_n$ satisfies  $F_n(0)=1$ and $F_n'(0)=F_n''(0)=0$ and there exist $0<A_1\leq A_2$, such that 
$$
A_1\leq |F_n(t)|\leq A_2,
$$
for $n \in \mathbb N$ and $t$ in a sufficiently small neighborhood of the origin.
\end{corollary}

\begin{remark} 
From Theorem \ref{thm:cumulants_bounded_and_decay}, we see that
\begin{equation}\label{eq:remark_not_true}
D_n(e^{tf})
=
e^{n t f_0} n^{t^2\gamma_n^2(f)} F_n(e^{tf}),
\end{equation}
where
\[
\gamma_n^2(f)
=
\frac{1}{\log n}
\sum_{k=1}^{\infty} \min(k,n)\, f_k f_{-k},
\]
and $F_n$ bounded from above and below.
In many situations, including the Fisher--Hartwig setting, the sequence
$\gamma_n^2(f)$ converges:
\[
\gamma_n^2(f)\to \gamma^2(f),
\qquad\text{as } n\to\infty.
\]
If the convergence is sufficiently fast, one may replace
$\gamma_n^2(f)$ by $\gamma^2(f)$ in \eqref{eq:remark_not_true}, giving a particularly elegant form. There are, however, (somewhat pathological) examples of functions $f$ satisfying
\eqref{eq:condition_on_fk} for which $\gamma_n^2(f)$ fails to
converge, even though the sequence is bounded. 
\end{remark}
We also mention a corollary to our results for the CUE. 

Suppose $f:\mathbb{T}\rightarrow\mathbb{R}$ satisfies $\|f_+\|_{\infty}, \|f_-\|_{\infty}< \infty$, and \footnote{The exact relation $|k|\,|f_k|=\kappa$ is imposed only for
simplicity of presentation; the assumption can be relaxed.} 
    $$
   |k| |f_k| =\kappa>0. 
    $$
Then we define a field on $[0,2\pi)$ by 
\begin{equation}\label{eq:field_def}
\Phi_n(\eta)=\Tr f(e^{-{i \eta}}U), \quad \eta \in [0,2 \pi).
\end{equation}
Similar to the discussion in Section~\ref{subsec:logcorfield}, one can show that the field $(\Phi_n(\eta))_\eta$ converges to a Gaussian log-correlated field, which is, again, a consequence of the Strong Szeg\H{o} Limit Theorem. Moreover,  Corollary \ref{cor:determinant_expansion} implies the following bound on the mixed exponential moments. 

\begin{corollary}\label{cor:mixed_exponential_moments}
    Suppose that $f:\mathbb T\to\mathbb R$ satisfies
\[ 
\|f_+\|_\infty,\ \|f_-\|_\infty<\infty,
\qquad
|k||f_k|=\kappa>0.
\]
Then, for all sufficiently small $s_1,s_2\in\mathbb R$,
\begin{multline}
\mathbb E\left[
\exp\left(
s_1\operatorname{Tr} f(e^{-i\eta_1}U)
+s_2\operatorname{Tr} f(e^{-i\eta_2}U)
\right)
\right]\\
\asymp
\begin{cases}
n^{\kappa^2(s_1^2+s_2^2)}
|e^{i\eta_1}-e^{i\eta_2}|^{-2\kappa^2s_1s_2},
&
|e^{i\eta_1}-e^{i\eta_2}|>n^{-1},
\\[4pt]
n^{\kappa^2(s_1+s_2)^2},
&
|e^{i\eta_1}-e^{i\eta_2}|\le n^{-1}.
\end{cases}
\end{multline}
Here, the implicit constants in $\asymp$ are positive and independent of
$n,\eta_1,\eta_2$.
\end{corollary}
\begin{proof}
 Set $f^\eta(e^{i \theta})=f(e^{i(\theta-\eta)})$. Note that 
    $$
    f^\eta_k=e^{-ik \eta} f_k,
    $$
    and since $f$ is assumed to be real-valued, we also have $f_{-k}=\overline{f_{k}}$.
    
     Due to Corollary \ref{cor:determinant_expansion} to deduce the result, we only need to estimate
     $$
     \sum_{k=1}^\infty \min (k,n) g_k g_{-k},$$
     where $g_k$ are the Fourier coefficients of $g(z)=s_1 f^{\eta_1}+s_2 f^{\eta_2}. $

    We start by noting that
    \begin{multline}
    \sum_{k=1}^\infty \min(k,n) g_k g_{-k}= \sum_{k=1}^n kg_k g_{-k}+ \mathcal O(1)
    \\=(s_1^2+s_2^2)\sum_{k=1}^{n}k|f_k|^2+ s_1s_2\sum_{k=1}^n k|f_k|^2(e^{ik(\eta_1-\eta_2)}+e^{ik(\eta_2-\eta_1)})+\mathcal{O}(1),
    \end{multline}
    as $n\to \infty$
    and to conclude the result we verify that
    \begin{multline}
    \sum_{k=1}^n k|f_k|^2(e^{ik(\eta_1-\eta_2)} + e^{ik(\eta_2-\eta_1)})
    = 2\kappa^2
    \sum_{k=1}^{n}
    \frac{\cos k(\eta_2-\eta_1)}{k}
    \\
    =\begin{cases}
     -2\kappa^2 \log|e^{i \eta_1}-e^{i \eta_2}|+\mathcal{O}(1), &\text{for}\quad |e^{i \eta_1}-e^{i \eta_2}|> n^{-1},
     \\  2\kappa^2 \log n+ \mathcal{O}(1), & \text{for}\quad |e^{i \eta_1}-e^{i \eta_2}|\leq n^{-1}, \end{cases}
    \end{multline}
    as $n\to  \infty$. The implicit constant in all $\mathcal O(1)$ terms can be taken uniform in $\eta_1,\eta_2\in [0,2 \pi),$ since the estimate on the truncation of $\log|1-z|$ is uniform.
\end{proof}

\begin{remark}
For comparison, we mention that for the characteristic polynomial, we have $f(z)=\log|1-z|=\frac{1}{2}\left(\log_+(1-z)+\log_+(1-1/z)\right)$, where $\log_+(1-z)=-\sum_{k=1}^{\infty} z^k/k$. Thus, in this case we have $\kappa=\frac12$. 
\end{remark}
\begin{remark}
    Just as for the characteristic polynomials, it is natural to expect that this bound is key to prove that the measure $e^{s \Phi_n(\eta)}n^{-\kappa^2 s^2} d\eta$ converges to a GMC measure, for sufficiently small $s$. This is an interesting problem for which we believe Corollary \ref{cor:mixed_exponential_moments} is an important first step. 

    Additionally, in Corollary \ref{cor:mixed_exponential_moments}, the parameter $s$ has to be sufficiently small, but we expect this can be relaxed.  Our estimates, however, are not sufficiently strong to deduce this. 
\end{remark}

\subsection{A general bound on cumulants}\label{sec:general_bounds}
We now drop the bounded decomposition assumed in the previous section and impose only \eqref{eq:condition_on_fk}. Under this weaker assumption, $e^f$ need not belong to $L_1(\mathbb T)$. Nevertheless, the quantities $C_m^{(n)}(f)$ remain well defined, as shown in the following proposition.

\begin{proposition}\label{prop:analyticity_unbounded_repeated}
Suppose $f$ satisfies \eqref{eq:condition_on_fk}. Then, for $t$ in a sufficiently small neighborhood of $t=0$, we have that  $e^{tf} \in L_1(\mathbb T)$. Thus, for all $n\in\mathbb{N}$, $\det T_n(e^{tf})$ and $\log\det T_n(e^{tf})$ are well-defined analytic functions in a neighborhood of $t=0$. 
\end{proposition}
This proposition is proved in Section \ref{sec:unbounded_cumulants}. 

For many explicit estimates in this work, instead of \eqref{eq:condition_on_fk}, we will assume the equivalent, quantitative condition that for some $\kappa>0$, 
\begin{equation}\label{eq:condition_on_fk2}
|f_k|\leq \kappa |k|^{-1}, \quad\text{for}\quad k\not=0.
\end{equation}

Our results require one more technical assumption. 
\begin{definition}\label{def:nehari_complete}
We consider functions $f\in L^2(\mathbb{T})$ satisfying \eqref{eq:condition_on_fk2} for some $\kappa>0$ such that there exist constants $C>0$, $\kappa^*\geq \kappa>0$, and functions $g_+$, $g_-$ analytic inside/outside the unit circle and
\begin{enumerate} 
\item[i)] \label{cond:boundedness}
$\|f_++g_-\|_{\infty}, \|g_++f_-\|_{\infty}< C$
\item[ii)] \label{cond:fourier}
$|(g_{\pm})_k|\leq \frac{\kappa^*}{|k|}$ for $k\not=0$.
\end{enumerate}
We will refer to such symbols as \textbf{admissible}. 
\end{definition}

By Nehari's theorem (see Theorem~\ref{thm:nehari}), one can always find
functions $g_+$ and $g_-$ satisfying condition~(i) with
$C=\kappa\pi$; see Corollary~\ref{cor:bounded_completion} in
Section~\ref{sec:nehari_hilbert}. Condition~(ii), however, does not follow
directly, and, to the best of our knowledge, there are no general results in
the literature guaranteeing that it holds. Nevertheless, in all the examples
considered below, including symbols with Fisher--Hartwig-type singularities,
both conditions are satisfied with $C=\kappa\pi$ and
$\kappa^*=\kappa$.

Condition~(ii) is needed to obtain bounds on certain operators that arise in
the proof. Readers familiar with Nehari's theorem will recognize that a symbol
satisfying \eqref{eq:condition_on_fk2}  gives rise to a bounded Hankel
operator whose norm is controlled by $C$. We recall the relevant aspects of
this theory in Section~\ref{sec:nehari_hilbert}.

We now turn to our main result for unbounded functions $f$.

\begin{theorem}\label{thm:unbounded_cumulants_intro}
For any admissible function $f$ we have, for $n \in \mathbb N$, that
 \begin{equation}
        C_2^{(n)}(f)= \sum_{k=1}^\infty \min (k,n) f_kf_{-k},
    \end{equation}
    and, for $m\geq 3$
    \begin{equation}
        \left|C_m^{(n)}(f)\right|\leq a (cm)^m,
    \end{equation}
    for some constants $a,c>0$ depending on $\kappa^*$ and $C$, but not $n$ or $m.$
\end{theorem}
\begin{remark}
   In Theorem~\ref{thm:unbounded_cumulants} below, we provide explicit values for these constants. Since we do not expect those values to be optimal, we have deliberately left the constants unspecified here.
\end{remark}
Note that this bound is weaker  than the corresponding bound in the bounded case, since it grows super-exponentially in $m$. Consequently, we cannot sum over $m$ to obtain a bound for the determinant itself. Although it is perhaps natural that bounded symbols should admit slightly better estimates, the Fisher--Hartwig setting suggests that our bound is far from optimal and may be substantially improved. Whether such an improvement holds in full generality or requires additional assumptions remains an open problem.

Given Theorem \ref{thm:unbounded_cumulants_intro}, we obtain a central limit theorem for linear statistics.

\begin{corollary}\label{cor:clt_linear}
Suppose $f:\mathbb{T}\rightarrow \mathbb{R}$ is admissible. Then, in distribution,
\begin{equation} 
    \frac{X_n(f)-\mathbb E X_n(f)}{\sqrt{2\sum_{k=1}^\infty \min(k,n) |f_k|^2}} 
    \to Z \sim N\left(0,1\right),\quad\text{as}\quad n\rightarrow\infty.
\end{equation}
\end{corollary}
 This is a known result. If $f\in H^{\frac12}$  then this is the CLT from the Strong Szeg\H{o} Limit Theorem. If $f \notin H^{\frac 12}$, then the variance grows with $n$. Under the additional assumption that $f$ is bounded, Corollary \ref{cor:clt_linear} is essentially due to Costin-Lebowitz \cite{CL} and Soshnikov \cite{Sosh}, and also follows from the concentration inequality in \cite{BreuerDuits}. For unbounded $f$, it follows from the more general result in \cite[Theorem 4.1]{DiaconisEvans} (their CLT only assumes a slowly varying criterion for $\sum_{|j|\leq k}|j||f_j|^2$ which is satisfied in our setting). Finally, we note that for a different set of linear statistics, the variance appears in that ($n$-dependent) form in \cite{DuiJoh}.

\subsection{Overview of the proof}

We now outline the main strategy behind the proofs of our results. 

Our approach relies on three key ingredients:
\begin{enumerate}
\item a regularization of the Toeplitz determinant that represents it as the determinant of a finite section of a product of four exponentials of Toeplitz operators;
\item the Baker--Campbell--Hausdorff formula, which rewrites this product as a single exponential;
\item a general bound for determinants of finite sections of operator exponentials.
\end{enumerate}
We discuss each of these ingredients in turn.

\subsubsection{Regularizing the determinant}
The first step is not new and has been exploited in several works on determinants of Toeplitz operators. To the best of our knowledge, the idea was first proposed by Basor-Helton in \cite{BasHelt}. It was also used in a proof by Basor-Widom of the Borodin-Okounkov identity \cite{BasorWidom}.

We regard the infinite Toeplitz matrix associated with the symbol $e^{tf}$ as
an operator
\[
T(e^{tf})\colon \ell_2(\mathbb{N})\to\ell_2(\mathbb{N}).
\]
Let $P_n$ denote the orthogonal projection onto the first $n$ coordinates,
\begin{equation}\label{eqn:pnqn}
P_n(x_0,x_1,\ldots,x_{n-1},x_n,\ldots)^T
=
(x_0,x_1,\ldots,x_{n-1},0,\ldots)^T,
\qquad
Q_n=I-P_n.
\end{equation}
Then
\[
D_n(e^{tf})
=
\det\bigl(P_nT(e^{tf})P_n+Q_n\bigr).
\]
We write
$$
f=f_++f_-,
$$
where $f_+$ is analytic in the unit disc and $f_-$ is analytic in its
exterior. Under mild assumptions,
$$
e^f=e^{f_-}e^{f_+}
$$
is a Wiener--Hopf factorization, and
\begin{equation}\label{eqn:WH}
T(e^{tf})
=
T(e^{tf_-})T(e^{tf_+})
=
e^{tT(f_-)}e^{tT(f_+)}.
\end{equation}
Although the products of Toeplitz operators do not in general correspond to
the products of their symbols, they do in this case because $e^{tf_+}$ is
analytic and $e^{tf_-}$ is anti-analytic. Equivalently, the corresponding
Toeplitz operators are triangular.

Following \cite{BasHelt}, we use this triangular structure to regularize the
operator inside the determinant. This gives
\begin{align}
\det&\bigl(P_ne^{tT(f_-)}e^{tT(f_+)}P_n+Q_n\bigr)
=
\det\bigl(P_ne^{tT(f_+)}P_n+Q_n\bigr)
\nonumber\\
&\quad\times
\det\Bigl(
P_ne^{-tT(f_+)}P_ne^{tT(f_-)}e^{tT(f_+)}
P_ne^{-tT(f_-)}P_n+Q_n
\Bigr)
\nonumber\\
&\quad\times
\det\bigl(P_ne^{tT(f_-)}P_n+Q_n\bigr)
\nonumber\\
&=
e^{ntf_0}
\det\Bigl(
P_ne^{-tT(f_+)}e^{tT(f_-)}e^{tT(f_+)}
e^{-tT(f_-)}P_n+Q_n
\Bigr).
\label{eqn:studied_fredholm}
\end{align}

The advantage of this representation is that the multiplicative commutator
\[
e^{-tT(f_+)}e^{tT(f_-)}e^{tT(f_+)}e^{-tT(f_-)}
\]
has better regularity than its individual factors. If $f$ satisfies
\eqref{eq:boundedHhalfnorm}, then this operator is of the form $I+$ trace
class. Hence
\begin{multline}
\lim_{n\to\infty}
\det\Bigl(
P_ne^{-tT(f_+)}e^{tT(f_-)}e^{tT(f_+)}
e^{-tT(f_-)}P_n+Q_n
\Bigr)
\\
=
\det\Bigl(
e^{-tT(f_+)}e^{tT(f_-)}e^{tT(f_+)}e^{-tT(f_-)}
\Bigr).
\end{multline}
Evaluating the determinant on the right-hand side yields the Strong
Szeg\H{o} Limit Theorem.

For symbols with Fisher--Hartwig singularities and, more generally, for symbols
satisfying only \eqref{eq:condition_on_fk}, condition
\eqref{eq:boundedHhalfnorm} may fail. In this case, the preceding
regularization is not sufficient by itself.

\subsubsection{The Baker--Campbell--Hausdorff expansion}

The next step in our approach is to continue the argument using the
Baker--Campbell--Hausdorff formula. It allows us to write
\begin{equation}\label{eqn:ztconcrete}
e^{-tT(f_+)}e^{tT(f_-)}e^{tT(f_+)}e^{-tT(f_-)}
=
e^{Z(t)},
\end{equation}
where
\begin{multline}\label{eqn:qualitative_BCH}
Z(t)
=
-t^2[T(f_+),T(f_-)]
\\
+\frac{t^3}{2}
[T(f_+)-T(f_-),[T(f_+),T(f_-)]]
+\cdots
=
\sum_{m=2}^{\infty}t^m\mathbf{K}_m^{\mathrm{(sum)}}.
\end{multline}
Here $\mathbf{K}_m^{\mathrm{(sum)}}$ is a linear combination of $m$-fold
nested commutators of $T(f_+)$ and $T(f_-)$. We refer to
Section~\ref{sec:bch} for a more detailed discussion.

Under the condition \eqref{eq:boundedHhalfnorm}, the commutator
$[T(f_+),T(f_-)]$ is trace class. Moreover, the traces of all higher-order
nested commutators vanish by cyclicity. Therefore,
\[
\Tr Z(t)
=
-t^2\Tr[T(f_+),T(f_-)]
=
t^2\sum_{k=1}^{\infty}k f_{k}f_{-k},
\]
which recovers the Strong Szeg\H{o} Limit Theorem.

If \eqref{eq:boundedHhalfnorm} fails, however, the commutator
$[T(f_+),T(f_-)]$ need not even be compact. The higher-order nested
commutators must then also be taken into account.

\subsubsection{A general determinant inequality}

The third and final ingredient is an estimate from \cite{BreuerDuits}.
Roughly speaking, if $B$ is a bounded operator on $\ell_2(\mathbb{N})$, then
there exists a constant $c>0$ such that
\[
\left|
\log\det\bigl(P_ne^{tB}P_n+Q_n\bigr)
-
t\Tr(P_n BP_n)
\right|
\leq
c\,
t^2\|P_nBQ_n\|_2
\|Q_nBP_n\|_2,
\]
for sufficiently small $t$, where $\|\cdot\|_2$ denotes the
Hilbert--Schmidt norm (cf~\ref{thm:generalthm}).

Applying this estimate with $tB=Z(t)$ reduces the problem to estimating
$\Tr P_nZ(t)P_n$ and controlling the Hilbert--Schmidt norms of the off-diagonal
corners
\[
\|P_n\mathbf{K}_m^{\mathrm{(sum)}}Q_n\|_2
\qquad\text{and}\qquad
\|Q_n\mathbf{K}_m^{\mathrm{(sum)}}P_n\|_2.
\]

The key point is that, although the operators
$\mathbf{K}_m^{\mathrm{(sum)}}$ need not be compact, their off-diagonal
corners are Hilbert--Schmidt, with norms uniformly bounded in $n$. These
bounds are established in
Theorems ~\ref{thm:bound_on_easy_nested} and
~\ref{thm:bound_on_difficult_nested}.

\subsection{Overview of the rest of the paper}

The remainder of the paper is organized as follows.  We start in Section \ref{subsec:examples} with examples of $f$ to which our main results apply. Then, in Section~\ref{sec:prel}, we review the results on Toeplitz matrices needed in the sequel and recall the Baker--Campbell--Hausdorff formula. In Section~\ref{sec:norm_estimates}, we study nested commutators of Toeplitz operators under the assumption that $f_\pm$ are bounded. The general case is treated in Section~\ref{sec:unbounded_bounds}. Finally, in Section~\ref{sec:proofs}, we combine these ingredients to prove our main results.

\section{Examples} \label{subsec:examples}
We now proceed by briefly discussing several examples.
\subsection{Examples of bounded functions}
\subsubsection{A classical family}
A rich family of continuous functions satisfying \eqref{eq:condition_on_fk} is obtained by introducing nonlinear oscillations into the classical Fourier series for the logarithm. For $\gamma \geq 0$ and $\rho \neq 0$, define \begin{equation}\label{eqn:gamma_examples} f_{\gamma,\rho}^{(1)}(e^{i\theta}) = \sum_{k=1}^{\infty} \frac{\cos\bigl(k\theta+\rho k^{\gamma}\bigr)}{k}, \qquad f_{\gamma,\rho}^{(2)}(e^{i\theta}) = \sum_{k=1}^{\infty} \frac{\sin\bigl(k\theta+\rho k^{\gamma}\bigr)}{k}. \end{equation} For $\gamma=0$ or $\gamma=1$, these series reduce, up to translations and fixed real linear combinations, to the classical Fourier series of a logarithmic singularity and the corresponding argument function. Thus, the functions in \eqref{eqn:gamma_examples} may be viewed as nonlinear-phase perturbations of the classical Fisher--Hartwig functions. 

When $0<\gamma<1$, the series in \eqref{eqn:gamma_examples} converge uniformly and therefore define continuous functions on $\mathbb{T}$. Nevertheless, the point $\theta=0$ remains distinguished, and the limiting functions exhibit nontrivial local behavior there. This follows, for example, from Theorem~5.2 in \cite{Zyg}; see also Figures~\ref{fig:first_cos} and~\ref{fig:first_sin}. 

When $\gamma>1$ and $\gamma\notin\mathbb{N}$, the increasingly rapid oscillations generated by the phase $k^\gamma$ again imply uniform convergence. This follows from higher-order van der Corput estimates for oscillatory exponential sums, combined with summation by parts; see, for example, \cite[Chapter~3]{MontgomeryTenLectures}. In this regime, the distinguished local behavior present for $0<\gamma<1$ is no longer apparent; see Figures~\ref{fig:second_cos} and~\ref{fig:second_sin}.

For integer $\gamma\geq 2$, the conclusion depends on the arithmetic properties of $\rho$. If $\frac{\rho}{2\pi}\in\mathbb{Q},$ then the phase factors $e^{i\rho k^\gamma}$ are periodic in $k$, and the resulting functions may have logarithmic singularities or jump discontinuities. If, on the other hand, $\frac{\rho}{2\pi}\notin\mathbb{Q},$  and $\frac{\rho}{2\pi}$ is of finite Diophantine type, then Weyl-sum estimates, together with summation by parts, imply uniform convergence of the series in \eqref{eqn:gamma_examples}; see, for example, \cite[Chapter~3]{MontgomeryTenLectures}. Their limits are therefore continuous and, in particular, bounded. 

Finally, the phase perturbation does not change the absolute values of the Fourier coefficients. More precisely, for every $k\neq 0$ and $j=1,2$, \begin{equation*} \left| \bigl(f_{\gamma,\rho}^{(j)}\bigr)_k \right| = \frac{1}{2|k|}. \end{equation*} Thus, although the oscillatory phase may turn the singular Fisher--Hartwig functions into continuous functions, it preserves the magnitude of their Fourier coefficients. Consequently, the associated fields in \eqref{eq:field_def} remain asymptotically log-correlated, and a corresponding bound on mixed exponential moments follows from Corollary~\ref{cor:mixed_exponential_moments}.

\begin{figure} 
    \centering
    \begin{subfigure}[b]
    {0.45\textwidth}
        \centering
        \includegraphics[width=\textwidth]{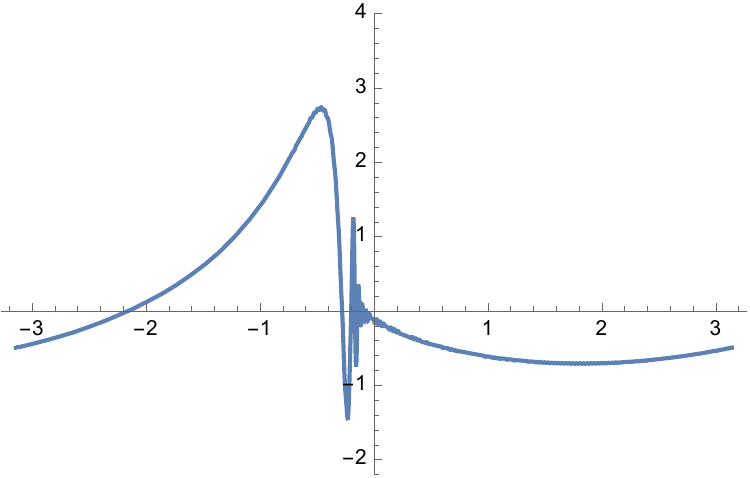}
        \caption{$f_\gamma^{(1)}$ with $\gamma=0.7$ }
        \label{fig:first_cos}
    \end{subfigure}
    \begin{subfigure}[b]{0.45\textwidth}
     \includegraphics[width=\textwidth]{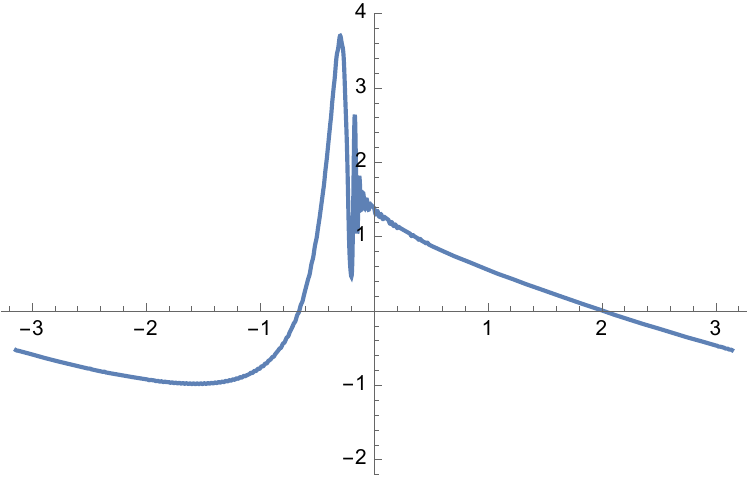}
        \caption{$f_\gamma^{(2)}$ with $\gamma=0.7$ }
        \label{fig:first_sin}
    \end{subfigure}
    \begin{subfigure}[b]{0.45\textwidth}
     \includegraphics[width=\textwidth]{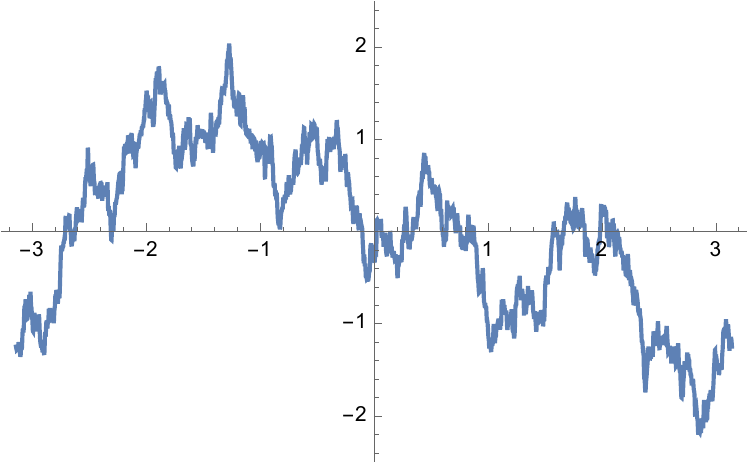}
        \caption{$f_\gamma^{(1)}$ with $\gamma=1.7$ }
        \label{fig:second_cos}
    \end{subfigure}
    \begin{subfigure}[b]{0.45\textwidth}
     \includegraphics[width=\textwidth]{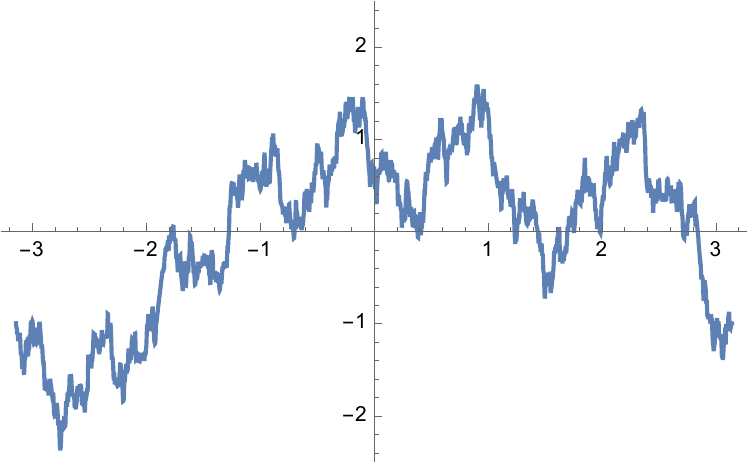}
        \caption{$f_\gamma^{(2)}$ with $\gamma=1.7$ }
        \label{fig:second_sin}
    \end{subfigure}
               \caption{Graphs of the partial sums of $f_\gamma^{(j)}$ with 200 terms}
               \label{fig:cos_sin_gamma}
\end{figure}
\subsubsection{Random Fourier coefficients}
When Fisher--Hartwig singularities are taken as the guiding examples, the assumption that $f_{\pm}$  are bounded may appear restrictive. In a certain sense, however, boundedness is the typical situation, whereas Fisher--Hartwig singularities are exceptional. Indeed, under mild assumptions, a Fourier series satisfying \eqref{eq:condition_on_fk}, whose Fourier coefficients are centered and chosen independently at random, defines an essentially bounded symbol almost surely; see \cite{Tal}. We recall the following elegant characterization.
\begin{theorem}[Talagrand, \cite{Tal}]\label{theorem:talagrand}
Given an i.i.d. mean zero
sequence of symmetric random variables,
$(X_n)_{n\in\mathbb{N}}$,
the random Fourier series
$$
\sum_{n\geq 1}\frac{X_n}{n}e^{in\theta}
$$
converges uniformly almost surely if and only if $\mathbb{E}[|X|LL(X)]<\infty$, where $LL(X)=\log\log(\max(e^e, |X|))$.
\end{theorem}
Consequently, the \textit{Rademacher series}, arising from $X_n\sim Ber(1/2)$, defines a random symbol, analytic inside the unit circle, and by considering its real and imaginary parts we obtain random functions fitting in the setting of Corollary \ref{cor:mixed_exponential_moments}. 

In a similar fashion, from Theorem \ref{theorem:talagrand} it follows that if we set $\xi_n\sim U(\mathbb{T})$ (by considering $X_n^{(1)}=\Re \xi_n$, $X_n^{(2)}= \Im \xi_n$) the series $\xi(\theta)=\sum_{n\geq 1}\frac{\xi_n}{n}e^{in\theta}$ produces another family of bounded symbols satisfying the setup of Corollary \ref{cor:mixed_exponential_moments} with complex Fourier coefficients. 

\subsection{Examples of admissible symbols}
\subsubsection{A general class of symbols}
A natural broad subfamily of symbols which are admissible (recall Definition \ref{def:nehari_complete}) is described as follows:
\begin{proposition}\label{ex:strip_support}
The linear span of functions 
$f_{\pm}:\mathbb{T}\rightarrow\mathbb{C}$
satisfying \eqref{eq:condition_on_fk} for 
which there exists a unit constant $e^{i\phi}$ such that 
$$\|\Im(e^{i\phi}f_{\pm})\|_{\infty}<\infty$$
is admissible. Equivalently $f(\mathbb{T})\subset \mathbf{S}$ for some strip $\mathbf{S}\subset\mathbb{C}$.
\end{proposition}
\begin{proof}
Since \eqref{eq:condition_on_fk} is linear on the Fourier coefficients, we can restrict our attention to a linear term, $f_{\pm}$. For readability, we consider $f_+$ (but the other case is identical). See that by assumption, for some $\phi\in[0,2\pi)$, and $C>0$,
\begin{equation}
|\Im(e^{i\phi}f_+(z))|
=
\left|\frac{1}{2i}(e^{i\phi}f_+(z)-e^{-i\phi}\overline{f_+(z)})\right|
=\frac{1}{2}\left| f_+(z)- e^{-i2\phi}\overline{f_+(z)} \right|
< C
\end{equation}
Thus we take $g_-=-e^{-i2\phi}\overline{f_+(z)}$.  $f_-$ is treated identically. We extend  linearly. 

\end{proof}

In particular,
the exponents of Fisher-Hartwig symbols, (recall \eqref{eqn:fh_log}) are contained in this class. Indeed, this observation is just a matter of checking that 
$\Im \log(1-z/\xi)=\arg (1-z/\xi)\in [-\frac{\pi}{2}, \frac{\pi}{2}]$.
We remark that in this case, recalling the assumed notation in Definition \ref{def:nehari_complete}, $\kappa=\kappa^*$ and $C\leq \kappa\pi$. 
\subsubsection{Examples with a dense set of singularities}
We conclude our list of examples with an admissible $f:\mathbb{T}\rightarrow \mathbb{C}$ which satisfies \eqref{eq:condition_on_fk}, but is unbounded on a dense set; there is no isolated singularity.

\begin{example}\label{ex:long_example}
Let $\theta_*/\pi$ be irrational and set
$$
\theta_j=j\theta_* \pmod{2\pi}.
$$
For $\alpha,\beta>0$, define
$$
f(z)
=
\sum_{j=1}^{\infty}\frac{1}{2^{j-1}}
\left(
\alpha \operatorname{Re}\log(1-ze^{-i\theta_j})
+
\beta \operatorname{Im}\log(1-ze^{-i\theta_j})
\right).
$$
Then $f$ is unbounded on the dense set
$\{e^{i\theta_j}\}_{j\geq 1}$, so its singularities are not isolated.
\end{example}

Nevertheless, here too, $\kappa^*=\kappa$ and $C\leq \kappa\pi$ (recall Definition \ref{def:nehari_complete}). We describe the
subleading asymptotics for $\alpha=1$ and $\beta=0$. We verify that
$$
f_+(z)
=
-\sum_{s=1}^{\infty}\frac{z^s}{s}
\sum_{j=1}^{\infty}\frac{e^{-ijs\theta_*}}{2^j}
=
-\sum_{s=1}^{\infty}
\frac{z^s}{s(2e^{is\theta_*}-1)},
$$
and it follows that
\begin{equation}\label{eqn:fssquared_expression}
|f_s|^2
=
\frac{1}{s^2}
\left(
\frac{1}{3}
+
\frac{2}{3}
\sum_{m=1}^{\infty}
\frac{\cos(ms\theta_*)}{2^m}
\right).
\end{equation}
Denoting
\[
L_n(\theta)
=
\sum_{s=1}^{n}\frac{\cos(s\theta)}{s},
\]
and observing that
\[
n\sum_{s=n+1}^{\infty}\frac{\cos(s\theta)}{s^2}
=
\mathcal{O}(1),
\]
uniformly in $\theta$, from \eqref{eqn:fssquared_expression}, by exchanging the order of summation, we obtain: 
\begin{equation}\label{eqn:infinite_symbol_norm}
\sum_{s=1}^{\infty}\min(s,n)|f_s|^2
=
\frac{1}{3}\log n
+
\frac{2}{3}
\sum_{m=1}^{\infty}\frac{L_n(m\theta_*)}{2^m}
+
\mathcal{O}(1),
\end{equation}
as $n\rightarrow\infty$. Moreover, since, uniformly in $\theta$, 
\begin{equation}\label{eqn:log_estimate}
L_n(\theta)
=
\log\min\left\{
n,\frac{1}{|1-e^{i\theta}|}
\right\}
+
\mathcal{O}(1),\quad\text{as}\quad n\rightarrow\infty,
\end{equation}
the second term in
\eqref{eqn:infinite_symbol_norm} approximates
\begin{equation}\label{eqn:seriesfinitefordiop}
-\frac{2}{3}
\sum_{m=1}^{\infty}
\frac{\log|1-e^{im\theta_*}|}{2^m},
\end{equation}
analogously to the mixed-angle
contribution in the Fisher--Hartwig formula.
This series is finite for typical irrational $\theta_*$, but may diverge
for suitably chosen \emph{Liouville numbers} (irrationals very well-approximated by rationals).

For an explicit example, we define
\begin{equation}\label{eqn:theta_star}
n_1=1,
\qquad
n_{k+1}=n_k+2^{2^{n_k}},
\qquad
\frac{\theta_*}{2\pi}
=
\sum_{k=1}^{\infty}2^{-n_k},
\end{equation}
and put
\[
m_k=2^{n_k}.
\]
The binary expansion gives that, as $k\rightarrow\infty,$
\[
|1-e^{im_k\theta_*}|
=
2\pi\,2^{n_k-n_{k+1}}(1+o(1)),
\]
and hence, uniformly
\begin{equation}\label{eqn:spike_estimate}
-\frac{1}{2^{m_k}}
\log|1-e^{im_k\theta_*}|
=
\log 2+\mathcal{O}(2^{-m_k}).
\end{equation}
Therefore, the  series  \eqref{eqn:seriesfinitefordiop} diverges through the sparse
sequence $(m_k)_{k\in\mathbb{N}}$.

We continue estimating the growth of the second term in \eqref{eqn:infinite_symbol_norm} (for $\theta_*$ as in \eqref{eqn:theta_star}). To this end, 
we define the (approximate) inverse of the power tower function by
\[
\log_2^* n
=
\#\left\{
k\geq 1:
|1-e^{im_k\theta_*}|^{-1}\leq n
\right\}.
\]
\begin{proposition}
For the angle $\theta_*$ defined in \eqref{eqn:theta_star}, the symbol in
Example~\ref{ex:long_example} with $\alpha=1$, $\beta=0$ satisfies
$$
\sum_{s=1}^{\infty}\min(s,n)|f_s|^2
=
\frac{1}{3}\log n
+
\frac{2}{3}\log 2\,\log_2^* n
+
\mathcal{O}(1),
\quad\text{as}\quad  n\to\infty.
$$
\end{proposition}

\begin{proof}
By \eqref{eqn:log_estimate} and \eqref{eqn:spike_estimate}, 
$$
\frac{L_n(m_k\theta_*)}{2^{m_k}}
=
\min\left\{
\frac{\log n}{2^{m_k}},
\log 2
\right\}
+
\mathcal{O}(2^{-m_k}).
$$
It follows that, as $n\rightarrow\infty$,
$$
\sum_{k=1}^{\infty}
\frac{L_n(m_k\theta_*)}{2^{m_k}}
=
\log 2\,\log_2^* n+\mathcal{O}(1).
$$
Elementary estimates based on the dyadic construction show that the
remaining frequencies contribute only a bounded amount:
$$
\sum_{m\notin\{m_k:k\geq 1\}}
\frac{L_n(m\theta_*)}{2^m}
=
\mathcal{O}(1),\quad\text{as}\quad n\rightarrow\infty.
$$
The result now follows from \eqref{eqn:infinite_symbol_norm}.
\end{proof}

This behavior is exceptional for irrationals that are extremely well approximated by rationals. Indeed,  By Khinchin's theorem, for almost every
irrational $\theta_*/(2\pi)$ and every $\varepsilon>0$, there exists
$C(\theta_*,\varepsilon)>0$ such that
$$
\|m\theta_*\|
\geq
C(\theta_*,\varepsilon)m^{-1-\varepsilon},
$$
where
$$
\|\theta\|
=
\min_{k\in\mathbb{Z}}|\theta-2\pi k|
$$
denotes the distance to $0$ on the circle. It follows that
$$
L_n(m\theta_*)=\mathcal{O}(\log m)
$$
uniformly in $n$, and therefore, as $n\rightarrow\infty,$
$$
\sum_{s=1}^{\infty}\min(s,n)|f_s|^2
=
\frac{1}{3}\log n+\mathcal{O}(1)
$$
for almost every $\theta_*$.

\section{Preliminaries }\label{sec:prel}

\subsection{Toeplitz and Hankel matrices/operators}

We now introduce the Toeplitz theory needed for this work. We refer the reader to the beautiful book of B\"ottcher and Silbermann, \cite{BottSilb06}, for an extensive account of the topic, including the results presented in this section.

The \textit{Toeplitz matrix} associated to function $a(z)=\sum_{k\in \mathbb Z} a_kz^k$ is defined by
$$T(a)=(a_{j-k})_{j,k=1}^{\infty},$$
and the \textit{Hankel matrix} associated to $a(z)$ is defined by
$$H(a)=\left(a_{j+k-1}\right)_{j,k=1}^{\infty}.$$
In this setting we refer to $a(z)$ as a \emph{symbol}. 

For such a symbol $a$, we define (in a standard fashion)
$$a_+(z)=\sum_{j=0}^{\infty}a_j z^j,\quad a_-(z)=\sum_{j=1}^{\infty}a_{-j}z^{-j},\quad \widetilde{a}(z)=a(1/z).$$
The infinite matrices naturally define Toeplitz and Hankel operators
and the following are standard identities they satisfy: for symbols $a,b\in L^2(\mathbb{T})$,
\begin{equation}\label{eqn:top_of_product}
T(ab)=T(a)T(b)+H(a)H(\widetilde{b}),
\end{equation}
\begin{equation}\label{eqn:jan_of_product}
H(ab)=T(a)H(b)+H(a)T(\widetilde{b}).
\end{equation}

As is evident from \eqref{eqn:top_of_product} and \eqref{eqn:jan_of_product},
Toeplitz and Hankel operators are closely related to multiplication operators on spaces of functions and in fact they can be represented by a composition of a multiplication operator with a projection operator on Hardy spaces. 

Moreover the norms of these operators are described by the norms of their symbols in the corresponding space of functions.

\begin{theorem}[Brown/Halmos]\label{thm:brown_halmos}
For $a\in L^{\infty}$,
$T(a): \ell_2\rightarrow \ell_2$ satisfies
\begin{equation}
\|T(a)\|_{op}\leq \|a\|_{\infty},
\end{equation}
Moreover, the Toeplitz matrix with symbol $a=(a_n)_{n\in\mathbb{Z}}$ is bounded iff $(a_n)_{n\in\mathbb{Z}}$ are the Fourier coefficients of some function $a\in L^{\infty}$ and
\begin{equation}
\|T(a)\|_{op}= \operatorname*{ess\,sup}_{|z|=1} \{|a(z)|\}.
\end{equation}
\end{theorem}

The norm of Hankel operators is described by Nehari's Theorem. 

\begin{theorem}[Nehari]\label{thm:nehari}
Suppose we are given a Hankel matrix, $H(a)$ with symbol $a=(a_n)_{n\in\mathbb{Z}}$, $a\in\ell_2$. Then $H(a):\ell_2\rightarrow\ell_2$ generates a bounded operator if and only if there exists a bounded function $b(z)=\sum_{n\in\mathbb{Z}}b_n z^n$ on $\mathbb{T}$, such that 
\begin{equation}\label{eqn:agreeing_coefs}
a_n=b_n\quad\text{ for } n\geq 1.
\end{equation}
Moreover in that case, there is a minimal $a_*$ satisfying \eqref{eqn:agreeing_coefs} so that
\begin{equation}\label{eqn:replacement_nehari}
\|H(a)\|_{op}=\inf\{\|\psi\|_{\infty}:  \psi_n=a_n\quad\text{for}\quad n\geq 1\}=\|a_*\|_{\infty}.
\end{equation}
\end{theorem}
Thus, due to \eqref{eqn:replacement_nehari} it is in general true on $\ell_2$ that
\begin{equation}
\|a_*\|_{\infty}=\|H(a)\|_{op}\leq \|T(a)\|_{op}= \|a\|_{\infty}. 
\end{equation}

\subsection{Hilbert's inequality and admissibility}\label{sec:nehari_hilbert}
A useful tool to bound the norm of Hankel operators is Hilbert's inequality.
\begin{proposition}[Hilbert's inequality]\label{prop:Hilberts_inequality}  
Suppose $a=(a_n)_{n\in\mathbb{N}}$ and $b=(b_n)_{n\in\mathbb{N}}$ are sequences of positive numbers. Then 

\begin{equation}\label{eqn:hilberts_inequality}
\sum_{k,\ell=1}^{\infty}\frac{a_kb_{\ell}}{k+\ell-1}\leq \pi\left(\sum_{k=1}^{\infty}a_k^2\right)^{1/2}\left(\sum_{\ell=1}^{\infty}b_{\ell}^2\right)^{1/2}
\end{equation}
\end{proposition}

\begin{proof}
This is in fact a consequence of Nehari's theorem applied to $-\log(1-z)$. There is also a proof in the book, \cite{HLP}. 
\end{proof}
An immediate application of Hilbert's inequality is that the Hankel operators associated to symbols $f$, satisfying \eqref{eq:condition_on_fk}, are bounded. 
\begin{corollary}\label{cor:bounded_completion}
For all functions $f\in L^2(\mathbb{T})$ satisfying \eqref{eq:condition_on_fk2}, $H(f), H(\widetilde{f}):\ell_2\rightarrow \ell_2$ are bounded operators and there exist $g_+$ and $g_-$ so that 
$$\|f_++g_-\|_{\infty}, \|g_++f_-\|_{\infty}
\leq\kappa \pi.$$
\end{corollary}
\begin{proof}
By assumption $|H(f)_{k\ell}|=|f_{k+\ell-1}|\leq\frac{\kappa}{k+\ell-1}.$ Take $v\in\ell_2(\mathbb{N})$. Since 
\begin{equation}
\|H(f)v\|_{2}
=\sup_{w\in\ell_2}\frac{|\langle w,H(f)v\rangle |}{\|w\|}
\leq\sup_{w\in\ell_2}
    \frac{1}{\|w\|}\sum_{k,\ell}\frac
    {\kappa |w_k||v_\ell|}
    {k+\ell-1}
\leq \kappa\pi\|v\|,
\end{equation}
it follows that $\|H(f)\|_{op}\leq \kappa\pi$. By Nehari's Theorem we can furthermore choose $g_-$ so that $\|H(f)\|_{op}=\|f_++g_-\|_{\infty}$.
Applying the same argument to $H(\widetilde{f})$, we prove the rest of the statement.
\end{proof}

Equivalently symbols $f$ satisfying \eqref{eq:condition_on_fk} are bounded up to some function in the orthogonal complement of $H^2$. We use a concrete procedure to utilize such functions in the orthogonal complement in Section \ref{sec:word_structure}, see Proposition \ref{prop:structure_of_nehari_family} and Lemma \ref{lem:replacement}.

\subsection{Baker--Campbell--Hausdorff}\label{sec:bch}
We now introduce the Baker--Campbell--Hausdorff identity. We refer the reader to the classical work, \cite{Mag}, and the book, \cite{BonFul12}, for a modern overview on the subject.

\begin{definition}[Commutator]
Given two elements $A, B$ in some ring $\mathcal{R}$,
\textit{the commutator} or \textit{bracket-product} or \textit{Lie-product} of $A$ and $B$ is 
$$[A,B]=AB-BA.$$
Commutators can be nested and we assume the convention that nesting is from the right to the left, e.g. for $A,B, C\in\mathcal{R}$,
\begin{equation}
[A,B,C]=[A,[B,C]]=ABC-ACB-BCA+CBA.
\end{equation}
\end{definition}
For compactness of notation it is often convenient to drop the commas in the notation of higher nested commutator;
namely as is customary, for $\ell$ letters $w_1,w_2,\ldots,w_{\ell}$, setting $\mathbf{w}=w_1w_2\ldots w_{\ell}$, we denote
\begin{equation}\label{eqn:word_nesting}
[\mathbf{w}]:=[w_1w_2\ldots w_{\ell}]=[w_1,w_2,\ldots, w_{\ell}].
\end{equation}
Let $\mathcal{R}_0=\mathbb{C}[X_1,X_2]$ be the free associative ring generated over 
$\mathbb{C}$
and free variables 
$X_1,X_2$. 
Let $\mathcal{R}=\mathbb{C}^*[X_1,X_2]$ be the ring of all formal power series in $X_1$ and $X_2$. The \textit{Baker-Campbell-Hausdorff} expansion is the following algebraic assertion:
\begin{theorem}[Baker-Campbell-Hausdorff expansion]\label{thm:bch}
There exists $Z\in\mathcal{R}$ which solves
\begin{equation}\label{eqn:non-com-log}
e^Z=e^Xe^Y.
\end{equation}
Furthermore $Z(t)=\log e^{tX}e^{tY}$ is a \textit{Lie-element}, that is it can be written as a linear sum of $X,Y$ and their nested commutators.
Moreover, 
\begin{equation}\label{eqn:dynkinbch}
Z(t)=
\sum_{m=1}^{\infty}\frac{t^m}{m}
\sum_{j=1}^m\frac{(-1)^{j+1}}{j}\sum_{\substack{u_1+v_1+\ldots u_j+v_j=m\\u_i+v_i\geq 1}}
\frac{\left[X^{u_1} Y^{v_1}\ldots X^{u_j} Y^{v_j}\right]}{u_1!v_1!\ldots u_j!v_j!}.
\end{equation}
\end{theorem}

We write below the first few terms in the expansion of $Z$ explicitly. 
\begin{equation}\label{eqn:bch_liee}
Z=X+Y+\frac{1}{2}[X,Y]+\frac{1}{12}[X,[X,Y]]-\frac{1}{12}[Y,[X,Y]]-\frac{1}{24}[Y,[X,[X,Y]]]
\ldots
\end{equation}
This result is formal, in this form, due to Dynkin, \cite{Dyn}. To use it analytically we require that the series should converge in an appropriate sense. We will use a more general version of this identity. 

Let $\mathcal{R}$ be the ring
$\mathbb{C}[X_1, X_2,\ldots, X_k]$ and the corresponding ring of formal power series be $\mathbb{C}^*[X_1, X_2,\ldots, X_k]$.

\begin{theorem}\label{thm:k-commutators}
There exists $Z\in\mathbb{C}^*[X_1,X_2,\ldots, X_k]$, which solves 
$$e^Z=e^{X_1}e^{X_2}\ldots e^{X_k}.$$
Furthermore $Z(t)= \log e^{tX_1}e^{tX_2}\ldots e^{tX_k}$ is a Lie element, that is $Z(t)$ is a linear sum of $X_1, X_2, \ldots, X_k$ and their nested commutators. 
\end{theorem}
\begin{proof}
Iterative applications of Theorem \ref{thm:bch} yield this result.
\end{proof}

Then, due to the Dynkin-Specht-Wever identity, see Appendix \ref{appendix:bch}, the expansion in Theorem \ref{thm:k-commutators} takes an exact form.

\begin{corollary}\label{cor:k-commutators}
Let $Z(t)$ be as in Theorem \ref{thm:k-commutators}, then
\begin{multline}\label{eqn:commutator_expansion}
Z(t)
=\sum_{m=1}^{\infty}\frac{t^m}{m} 
\sum_{j=1}^{m}\frac{(-1)^{j+1}}{j}
\sum_{
\substack{
n_1+n_2+\ldots +n_j=m
\\n_1,n_2,\ldots,n_j\geq 1
}}
\\ 
\sum_{
\substack{
m_1^{(s)}+m_2^{(s)}+\ldots +m_k^{(s)}=n_s
\\
s=1,2,\ldots, j 
}
}
\frac
{[X_1^{m_1^{(1)}}\ldots X_k^{m_k^{(1)}}\ldots \ldots X_1^{m_1^{(j)}}\ldots X_k^{m_k^{(j)}}]}
{m_1^{(1)}! \ldots m_k^{(1)}!\ldots\ldots m_1^{(j)}! \ldots m_k^{(j)}!}.
\end{multline}
The sum of the absolute values of all coefficients in front of commutator words $[w]$ of length $m$ is bounded by $\frac{(ke)^m}{\sqrt{2\pi m^3}}$.
\end{corollary}

\begin{proof}
See Appendix \ref{appendix:bch}. 
\end{proof}

\begin{remark}
There is cancellation and repetition in \eqref{eqn:commutator_expansion} and even more so for the concrete setting in which we use Corollary \ref{cor:k-commutators}: $k=4$, $X_1=-X_3$, $X_2=-X_4$ (recall \eqref{eqn:ztconcrete}). Nevertheless, \eqref{eqn:commutator_expansion} is appealing for its simple and general expression for the terms indexed by $t^m$, which is amenable to analysis.
However, a better understanding of this cancellation, perhaps by studying carefully modern and efficient algorithms, see e.g.
\cite{VanVis}, could lead to better estimates in our main results. 
\end{remark}

\section{Bounds on nested commutators for bounded symbols}
\label{sec:norm_estimates}
We prove bounds on the nested commutators
\begin{equation}\label{eqn:m-nested_commutator}
    \mathbf{K}_m=[T(f^{(m)}),\ldots, T(f^{(2)}), T(f^{(1)})].
\end{equation}
These bounds are the main ingredients of the proof of Theorem \ref{thm:cumulants_bounded_and_decay}. In this section, we will work under the assumption that each symbol is (essentially) bounded. The unbounded case will be treated in the next section and will lead to weaker bounds. 

The main result we prove is the following.

\begin{theorem}\label{thm:bound_on_easy_nested}
Suppose for some $C>0$, and $m\geq 2$, that the symbols $f^{(1)},f^{(2)},\ldots, f^{(m)}:\mathbb{T}\rightarrow\mathbb{C}$ satisfy $\|f^{(k)}\|_{\infty}< C$ for $k=1,\ldots, m$ and assume further that they satisfy \eqref{eq:condition_on_fk2} with a common $\kappa>0$. Then
\begin{enumerate}
\item $\|\mathbf{K}_m\|_{op}\leq C^m 2^{m-1}$,
\item $\|P_n\mathbf{K}_mQ_n\|_{2}, \|Q_n\mathbf{K}_mP_n\|_{2}\leq 2^{m-1} 3\kappa C^{m-1} m \sqrt{\log(1+m)}$,
\item $\left| \Tr P_n \mathbf{K}_m P_n \right|
\leq 2^{m-1} 3 \kappa^2  C^{m-2} m^2 \log(1+m)
\quad\text{for}\quad m\geq 3.$
\end{enumerate}
\end{theorem}

The rest of this section is devoted to the proof of this theorem. 
\subsection{Bounds on words}
We will start with the following an expansion of each $m$-nested commutator in monomials.
\begin{proposition}\label{prop:structural_1}
For $m\geq 2$,
we have that 
\begin{enumerate}
\item $\mathbf{K}_2=H(f^{(1)})H(\widetilde{f^{(2)}})-H(f^{(2)})H(\widetilde{f^{(1)}})$,
\item $\mathbf{K}_m$ is a sum of $2^{m-1}$ terms of monomials of degree $m$ (multiplied by $+1$ or $-1$) of which 
$2^{m-2}$ correspond to permutations of $H(f^{(1)})$, $H(\widetilde{f^{(2)}})$, $T(f^{(3)}),\ldots, T(f^{(m)})$ and $2^{m-2}$ correspond to permutations of
$H(f^{(2)}), H(\widetilde{f^{(1)}}), T(f^{(3)}),\ldots, T(f^{(m)})$.
\end{enumerate}
\end{proposition}
\begin{proof}
This result is a straightforward consequence of \eqref{eqn:top_of_product}, but is stated separately for clarity.
\end{proof}

We will refer to the monomials of degree $\ell$ in the expansion in Proposition \ref{prop:structural_1} as words of length $\ell$ and denote them $\mathbf{W}_{\ell}$. 
For reasons related to the treatment of unbounded symbols we state the following results for a broader class of words. %

The first main proposition of this section gives us an estimate on the Hilbert-Schmidt norm of $P_n\mathbf{W}_{\ell}Q_n$ and $Q_n \mathbf{W}_{\ell} P_n$.
\begin{proposition}\label{prop:word_corner}
Suppose $f:\mathbb{T}\rightarrow\mathbb{C}$ satisfies $\|f\|_{\infty}<C$ and has Fourier coefficients $(f_k)_{k\in\mathbb{Z}}$ satisfying $|f_k|\leq \kappa|k|^{-1}$.
Suppose $\mathbf{A}_1, \mathbf{A}_2$ are monomials of Toeplitz and Hankel operators of symbols satisfying the same constraints as $f$ so that the total degree of 
$A_1 H(f) A_2$ is $\ell$.
Then
\begin{equation}
\|P_n \mathbf{A}_1 H(f)\mathbf{A}_2 Q_n\|_2
\leq 
3\kappa C^{\ell-1} \ell \sqrt{\log(1+\ell)}
\end{equation}
and also
\begin{equation}
\|Q_n \mathbf{A}_1 H(f)\mathbf{A}_2 P_n\|_2\leq 3\kappa C^{\ell-1} \ell \sqrt{\log(1+\ell)}.
\end{equation}
\end{proposition}

The second main proposition of this section gives an estimate on $\Tr P_n \mathbf{W}_{\ell} P_n$. 
\begin{proposition}\label{prop:word_trace}
Suppose, for $\ell\geq 3$,
\begin{equation}\label{eqn:form_5_operators}
\mathbf{W}_{\ell-1}=\mathbf{A}_1H(f^{(1)})\mathbf{A}_2 H(f^{(2)}) \mathbf{A}_3, 
\end{equation}
where $f^{(1)}, f^{(2)}$ satisfy \eqref{eq:condition_on_fk2} (with a common $\kappa>0$) and $\|f^{(1)}\|_{\infty}, \|f^{(2)}\|_{\infty}<C$ for some $C>0$, and
$\mathbf{A}_1, \mathbf{A}_2, \mathbf{A}_3$ are words of Toeplitz and Hankel operators of symbols satisfying the same constraints as $f^{(1)}$ and $f^{(2)}$.
Suppose $b$ is a symbol with Fourier coefficients 
satisfying \eqref{eq:condition_on_fk2} (for the same $\kappa>0$), and assume either that $\|b\|_{\infty}<C$, or that $b$ is admissible with underlying constants $\kappa^*=\kappa$ and $C$, or that $0<\kappa \pi\leq C$.
Then
\begin{equation}
| \Tr P_n[ \mathbf{W}_{\ell-1}, T(b)]P_n |
\leq 6 \kappa^2  C^{\ell-2} \ell^2 \log(1+\ell).
\end{equation}
\end{proposition}

To prove these two results, we first establish some norm estimates for Hankel and Toeplitz operators.

\subsection{Some first estimates}
We first prove several auxiliary estimates needed for the proof of Propositions \ref{prop:word_corner} and \ref{prop:word_trace}.
We begin with several estimates on the Hilbert-Schmidt norm of "$n$-corners" of Toeplitz and Hankel operators. 
\begin{lemma}\label{lemma:hslemma}
Suppose $f:\mathbb{T}\rightarrow \mathbb{C}$ has a Fourier expansion given by $f(z)=\sum_{k\in\mathbb{Z}}f_k z^k.$
Assume further that $|f_k|\leq \kappa |k|^{-1}$ for $k\not=0$ and some $\kappa>0$. 
For $ s\in \mathbb{N}$, we have that, 

\begin{equation} \label{item:hs1}
\|P_nH(f)Q_{[n/s]}\|_2\leq 
\kappa\sqrt{2\log(1+s)},
\end{equation}
\begin{equation}\label{item:hs2}
\|Q_nH(f) P_{sn}\|_2\leq \kappa\sqrt{\log(1+s)},
\end{equation}
\begin{equation} \label{item:hs3}
\|P_{[n/2s]} T(f) Q_{[n/s]}\|_2 \leq \kappa\sqrt{\log 2},
\end{equation}  
\begin{equation}\label{item:hs4}
 \|Q_{(s+1)n}T(f)P_{sn}\|_2\leq \kappa\sqrt{\log(1+s)}.
\end{equation}

\end{lemma}
\begin{proof}[Proof of Lemma \ref{lemma:hslemma}]
We consecutively verify the four estimates. For \eqref{item:hs1}, when $s\geq n+1$, we have that $Q_{[n/s]}=id$ and
$$\|P_n H(f)\|_2^2\leq \kappa^2\sum_{k=1}^{\infty} \frac{\min(k,n)}{k^2}\leq \kappa^2(2+\log n)\leq 2 \kappa^2\log(n+2).$$
Else when $s\leq n$,
\begin{multline}
\|P_nH(f)Q_{[n/s]}\|_2^2\leq\sum_{k=1}^{n}\sum_{j=1}^{\infty}\frac{\kappa^2}{([n/s]+j+k-1)^2}\leq \kappa^2\log\left(\frac{[n/s]+n}{[n/s]}\right)
\\
\leq \kappa^2\log
\left(
1+\frac{n}{[n/s]}
\right)
\leq\kappa^2\log(1+2s)<2\kappa^2\log(1+s).
\end{multline}
For \eqref{item:hs2},
\begin{multline}
\|Q_nH(f) P_{sn}\|_2^2\leq  \kappa^2 \sum_{k=1}^{sn}\sum_{j=1}^{\infty} \frac{1}{(n+k+j-1)^2}
\\
\leq \kappa^2\sum_{k=1}^{sn}\frac{1}{n+k-\frac{1}{2}} \leq \kappa^2 \log
\left( 
\frac{n+sn}{n}
\right)
= \kappa^2\log(s+1).
\end{multline}
For \eqref{item:hs3}, if $n<2s$, the left-hand side is zero. Otherwise, we have that
\begin{multline}
\|P_{[n/2s]} T(f) Q_{[n/s]}\|_2^2\leq
\kappa^2\sum_{\substack{k=\\ [\frac{n}{s}] -[\frac{n}{2s}]
}}
^{[\frac{n}{s}]-1}
\sum_{j=0}^{\infty} \frac{1}{(k+1+j)^2}
\leq
\kappa^2\sum_{\substack{k=\\ [\frac{n}{s}] -[\frac{n}{2s}]
}}
^{[\frac{n}{s}]-1}
\frac{1}{(k+1/2)}
\\
\leq \kappa^2 \log\left(\frac{[\frac{n}{s}]}{[\frac{n}{s}]-[\frac{n}{2s}]}\right) \leq \kappa^2\log 2.
\end{multline}
Finally, we observe that \eqref{item:hs4} is equivalent to \eqref{item:hs2}, 
\begin{equation}
\|Q_{(s+1)n}T(f)P_{sn}\|_2^2=\| Q_n H(f) P_{sn}\|_2^2\leq  \kappa^2\log(s+1).
\end{equation}

\end{proof}
Note that these results are valid for all symbols satisfying \eqref{eq:condition_on_fk}. In particular, they even hold for unbounded symbols. 
\subsection{Estimates on longer words}

We use Lemma \ref{lemma:hslemma} as a building block to treat longer words. It is illustrative to do a small example encapsulating the dyadic expansion we use to treat longer words.
\begin{lemma}\label{lemma:dyadic_example}
For symbols $f,g$, essentially bounded by $C>0$, satisfying \eqref{eq:condition_on_fk2}, we have that
\begin{equation}
\|P_nT(f)H(g)Q_n\|_2\leq 2C\kappa\sqrt{\log 3},\quad \|Q_nT(f)H(g)P_n\|_2\leq (1+\sqrt{2})C\kappa\sqrt{\log 3}.
\end{equation}

\end{lemma}
\begin{proof}
We begin with the first inequality. We utilize the following decomposition:
$$P_nT(f)H(g)Q_n=P_nT(f)Q_{2n}H(g)Q_n+P_nT(f)P_{2n}H(g)Q_n.$$
Although $P_nT(f)Q_n$ typically has a divergent Hilbert-Schmidt norm for the symbols we consider, by Lemma \ref{lemma:hslemma},  $P_nT(f)Q_{2n}$ has a bounded Hilbert-Schmidt norm. Similarly $P_{2n} H(g)Q_n$ has a uniformly bounded Hilbert-Schmidt norm. The boundedness of the symbols guarantees that these observations are sufficient:
\begin{multline}
\|P_nT(f)H(g)Q_n\|_2
\\
\leq 
\|P_nT(f)Q_{2n}\|_2\|H(g)Q_n\|_{op}+\|P_nT(f)\|_{op}\|P_{2n}H(g)Q_n\|_2
\\\leq \kappa\sqrt{\log 3}C+C\kappa\sqrt{\log 3}=2C\kappa\sqrt{\log 3}.
\end{multline}

Similarly, since
$$Q_nT(f)H(g)P_n
=
Q_nT(f)P_{[n/2]}H(g)P_n
+ Q_nT(f)Q_{[n/2]}H(g)P_n,$$
it follows that
\begin{multline}
\|Q_nT(f)H(g)P_n\|_2
\\
\leq \|Q_nT(f)P_{[n/2]}\|_2\|H(g)P_n\|_{op}
+ \|Q_nT(f)\|_{op}\|Q_{[n/2]}H(g)P_n\|_2
\\
\leq 
\kappa\sqrt{\log 2} C+ C\kappa\sqrt{2\log 3}\leq(1+\sqrt{2})C\kappa\sqrt{\log 3}.
\end{multline}
\end{proof}

We now iteratively apply the idea presented in Lemma \ref{lemma:dyadic_example} to prove the following estimate for the Hilbert-Schmidt norm of the "corners" of words of length $\ell>0$.
\begin{proposition}\label{prop:one_side}
Assume
$f^{(0)},f^{(1)}, f^{(2)}, \ldots, f^{(\ell)}$ are essentially bounded symbols (by a constant $C>0$) with Fourier coefficients 
satisfying \eqref{eq:condition_on_fk2} for a common $\kappa>0$. Then
$$\|P_n H(f^{(0)}) T(f^{(1)}) T(f^{(2)})\ldots T(f^{(\ell)})Q_n\|_2\leq \kappa C^{\ell}(2\ell+\sqrt 2)\sqrt{\log 2},$$
and more generally, for $s\in\mathbb{N}$,
\begin{equation}\label{eqn:pq_general}
\|P_n H(f^{(0)}) T(f^{(1)})\ldots T(f^{(\ell)})Q_{[n/s]}\|_2
\leq \kappa C^{\ell}\left(2\ell \sqrt{\log 2}+ \sqrt{2 \log(1+s)}\right).  
\end{equation}
\end{proposition}

\begin{proof}
We proceed to prove the general \eqref{eqn:pq_general} by induction on $\ell\geq 0$. 
\medskip
 \\For $\ell=0$, by Lemma \ref{lemma:hslemma}, \eqref{item:hs1}, we have that for $s\geq 1$,
\begin{equation}
\|P_nH(f^{(0)})Q_{[n/s]}\|_2 \leq \kappa\sqrt{2\log(1+s)}.
\end{equation}
For $\ell>0$, we assume inductively that \eqref{eqn:pq_general} holds for all integers $0\leq\ell'<\ell$. 
Then, if $s>n$, $Q_{[n/s]}=id$ and so
 \begin{multline}
 \|P_n H(f^{(0)}) T(f^{(1)}) T(f^{(2)})\ldots T(f^{(\ell)}) \|_2
\\
 \leq 
 \|P_nH(f^{(0)})\|_2
 \|T(f^{(1)}) T(f^{(2)})\ldots T(f^{(\ell)})\|_{op}\leq 
\kappa \sqrt{2\log(1+n)}C^{\ell}.
 \end{multline}
Else, if $s\leq n$, we have 
\begin{multline}
P_n H(f^{(0)}) T(f^{(1)}) T(f^{(2)})\ldots T(f^{(\ell)})Q_{[n/s]}
\\= P_n H(f^{(0)}) T(f^{(1)}) T(f^{(2)})\ldots T(f^{(\ell-1)})P_{[n/2s]} T(f^{(\ell)})Q_{[n/s]}
\\+
P_n H(f^{(0)}) T(f^{(1)}) T(f^{(2)})\ldots T(f^{(\ell-1)}) Q_{[n/2s]} T(f^{(\ell)})Q_{[n/s]}.
\end{multline}
By the inductive hypothesis, we have
\begin{multline}
\|P_n H(f^{(0)}) T(f^{(1)}) T(f^{(2)})\ldots T(f^{(\ell-1)}) Q_{[n/2s]}\|_2
\\\leq \kappa C^{\ell-1}\left(2(\ell-1)\sqrt{\log 2}+\sqrt{2\log(1+2s)} \right)
\\\leq \kappa C^{\ell-1} \left((2\ell-1)\sqrt{\log 2}+\sqrt{2\log(1+s)} \right),
\end{multline} 
and since $\|f_{\ell}\|_{\infty}<C$, it follows that 
$\|T(f_{\ell})\|_{op}, \| T(f_{\ell})Q_{[n/s]}\|_{op}<C$, and so
\begin{multline}\label{eqn:proofequation1}
\|P_n H(f^{(0)}) T(f^{(1)}) T(f^{(2)})\ldots T(f^{(\ell-1)})Q_{[n/2s]} T(f^{(\ell)})Q_{[n/s]}
\|_2
\\
\leq 
\|P_n H(f^{(0)}) T(f^{(1)}) T(f^{(2)})\ldots T(f^{(\ell-1)})Q_{[n/2s]}\|_2  
\|T(f^{(\ell)})Q_{[n/s]}\|_{op}
\\ 
\leq \kappa C^{\ell} \left((2\ell-1)\sqrt{\log 2}+\sqrt{2\log(1+s)} \right).
\end{multline}
For the other term, we make use of Lemma \ref{lemma:hslemma}, \eqref{item:hs3}, together with the boundedness of the Toeplitz operators:
\begin{multline} \label{eqn:proofequation2}
\|P_n H(f^{(0)}) T(f^{(1)}) T(f^{(2)})\ldots T(f^{(\ell-1)})P_{[n/2s]} T(f^{(\ell)})Q_{[n/s]}\|_2
\\
\leq 
\|P_n H(f^{(0)}) T(f^{(1)}) T(f^{(2)})\ldots T(f^{(\ell-1)})\|_{op}\|P_{[n/2s]} T(f^{(\ell)})Q_{[n/s]}\|_2
\\
\leq
 C^{\ell} \kappa \sqrt{\log 2}.
\end{multline}

Summing \eqref{eqn:proofequation1} and \eqref{eqn:proofequation2} we obtain
\begin{multline}
\| P_n H(f^{(0)}) T(f^{(1)}) T(f^{(2)})\ldots T(f^{(\ell-1)}) T(f^{(\ell)})Q_{[n/s]} \|_2
\\\leq \kappa C^{\ell} \left(2\ell \sqrt{\log 2}+\sqrt{2\log(1+s)} \right).
\end{multline}
In particular, for $s=1$ we obtain the other part of the statement.
\end{proof}

We need to prove a corresponding statement when the projections $P_n$ and $Q_n$ are in the opposite order.

\begin{proposition}\label{prop:other_side}
Let the symbols $f^{(0)},f^{(1)},\ldots, f^{(\ell)}$ be as in the setting of Proposition \ref{prop:one_side}. Then
$$
\|Q_n H(f^{(0)}) T(f^{(1)}) T(f^{(2)})\ldots T(f^{(\ell)}) P_n\|_2\leq \kappa C^{\ell} (\ell+1) \sqrt{\log(2+\ell)},
$$
and more generally for $s\in\mathbb{N}$
\begin{multline}\label{eqn:dyadic_qn_pn_general}
\|Q_n H(f^{(0)}) T(f^{(1)}) T(f^{(2)})\ldots T(f^{(\ell)}) P_{sn}\|_2\leq  \kappa C^{\ell}\sum_{m=0}^{\ell}\sqrt{\log(1+s+m)} 
\\\leq \kappa C^{\ell} (\ell+1)\sqrt{\log(1+s+\ell)}.
\end{multline}
\end{proposition}

\begin{proof}
We proceed by induction on $\ell\geq 0$. We prove the more general hypothesis given in \eqref{eqn:dyadic_qn_pn_general}.
%
\medskip
\\Once again for $\ell=0$ and
for $s\in\mathbb{N}$, by Lemma \ref{lemma:hslemma}, \eqref{item:hs2}, we have that
$\|Q_nH(f^{(0)}) P_{sn}\|_2\leq  
 \kappa\sqrt{\log(s+1)}.
$
\medskip
\\For $\ell>0$, we assume inductively that \eqref{eqn:dyadic_qn_pn_general} holds for all integers $0\leq\ell'<\ell$. 
Then, for $s\in\mathbb{N}$, we have
\begin{multline}
Q_n H(f^{(0)}) T(f^{(1)}) T(f^{(2)})\ldots T(f^{(\ell)})P_{sn}
\\= Q_n H(f^{(0)}) T(f^{(1)}) T(f^{(2)})\ldots T(f^{(\ell-1)})P_{(s+1)n} T(f^{(\ell)})P_{sn}
\\+
 Q_n H(f^{(0)}) T(f^{(1)}) T(f^{(2)})\ldots T(f^{(\ell-1)})Q_{(s+1)n} T(f^{(\ell)})P_{sn}.
\end{multline}
For the first summand we have that
\begin{multline}
 \|Q_n H(f^{(0)}) T(f^{(1)}) T(f^{(2)})\ldots T(f^{(\ell-1)})P_{(s+1)n} T(f^{(\ell)})P_{sn}\|_2
\\\leq 
 \|Q_n H(f^{(0)}) T(f^{(1)}) T(f^{(2)})\ldots T(f^{(\ell-1)}) P_{(s+1)n}\|_2 \|P_{(s+1)n}T(f^{(\ell)})P_{sn}\|_{op}
 \\ \leq  \kappa C^{\ell-1} \sum_{m=0}^{\ell-1}\sqrt{\log(2+s+m)}\times C
 =\kappa C^{\ell} \sum_{m=1}^{\ell}\sqrt{\log(1+s+m)},
\end{multline}
where we used the boundedness of $f^{(\ell)}$ and the inductive hypothesis in the last inequality. For the second summand we have that:
\begin{multline}
\|Q_n H(f^{(0)}) T(f^{(1)}) T(f^{(2)})\ldots T(f^{(\ell-1)}) Q_{(s+1)n} T(f^{(\ell)})P_{sn}\|_2
\\ \leq
 \|Q_n H(f^{(0)}) T(f^{(1)}) T(f^{(2)})\ldots T(f^{(\ell-1)})  Q_{(s+1)n}\|_{op}
\|Q_{(s+1)n}T(f^{(\ell)})P_{sn}\|_2
\\ \leq
 C^{\ell}  \|Q_{(s+1)n}T(f^{(\ell)})P_{sn}\|_2
 \leq \kappa C^{\ell} \sqrt{\log(s+1)},
\end{multline}
where the last inequality is justified by Lemma \ref{lemma:hslemma}, \eqref{item:hs4}.
Summing the two bounds gives us the bound from the inductive hypothesis.
\end{proof}

More generally, we observe that if we replace any $T(f^{(m)})$ with $H(f^{(m)})$ or any of these two endowed with projections $P_s$ or $Q_s$ for some $s\geq 1$, the estimates in the proof of Propositions \ref{prop:one_side} and \ref{prop:other_side} remain valid. Thus we have the following result. 

\begin{corollary} \label{cor:important_estimate}
Assume that 
$f^{(0)},f^{(1)}, f^{(2)}, \ldots, f^{(\ell)}$
have
Fourier coefficients 
satisfying \eqref{eq:condition_on_fk2} for a common $\kappa>0$.
Suppose further that 
$B(f^{(k)})= T(f^{(k)})\quad \text{or}\quad H(f^{(k)})$, or one of these operators applied together with one of the projections $Q_s$ or $P_s$ on the left or/and on the right, and assume $\|B(f^{(k)})\|_{op}\leq C$, for $k=0,\ldots,\ell$ for some $C>0$.
Then we have: 
\begin{equation}
\|Q_n H(f^{(0)}) B(f^{(1)}) B(f^{(2)})\ldots B(f^{(\ell)})P_n\|_2\leq 2 \kappa C^{\ell} (\ell+1) \sqrt{\log(2+\ell)},
\end{equation} 
\begin{equation}
\|P_n H(f^{(0)}) B(f^{(1)}) B(f^{(2)})\ldots B(f^{(\ell)})Q_n\|_2\leq 2 \kappa C^{\ell} (\ell+1) \sqrt{\log(2+\ell)},
\end{equation}
\begin{equation}
\|Q_n B(f^{(1)}) B(f^{(2)})\ldots B(f^{(\ell)})H(f^{(0)})P_n\|_2\leq 2 \kappa C^{\ell} (\ell+1) \sqrt{\log(2+\ell)},
\end{equation} 
\begin{equation}
\|P_n   B(f^{(1)}) B(f^{(2)})\ldots B(f^{(\ell)})H(f^{(0)}) Q_n\|_2\leq 2 \kappa C^{\ell} (\ell+1) \sqrt{\log(2+\ell)}.
\end{equation}
\end{corollary}

\begin{proof}
The proofs of the first and second inequality follow line for line correspondingly the proofs of Propositions \ref{prop:one_side} and \ref{prop:other_side}.  Namely in these proofs when replacing $T(f^{(k)})$ with any of the other operators in the given, the same bounds apply to the estimated operator and Hilbert-Schmidt norms in the proofs. Finally the bound appearing on the right-hand side above is chosen to be greater than the two corresponding bounds appearing in Propositions \ref{prop:one_side} and \ref{prop:other_side}.
The latter two inequalities can be viewed as estimates on the Hilbert-Schmidt norm of the transpose of the left-hand side in the first two estimates. 
\end{proof}

\begin{remark}\label{rmk:important_estimate}
For the conditions of Corollary \ref{cor:important_estimate} to be satisfied, apart from \eqref{eq:condition_on_fk2}, it suffices to additionally assume $\|f^{(k)}\|_{\infty}\leq C$ for $k=0,\ldots, \ell$. However, this is not strictly necessary. When the operator
$$B(f^{(k)})=H(f^{(k)}),\quad P_nT(f^{(k)})Q_n, \quad Q_nT(f^{(k)})P_n,$$
it depends only on the positive or only on the negative Fourier coefficients of $f^{(k)}$ and so,
due to Corollary \ref{cor:bounded_completion}, condition \eqref{eq:condition_on_fk2} alone implies $\|B(f_k)\|_{op}\leq\kappa\pi$; and admissibility implies $\|B(f_k)\|_{op}\leq C$.
\end{remark}

\subsection{Proofs of Propositions \ref{prop:word_corner} and \ref{prop:word_trace}}
With Corollary \ref{cor:important_estimate} at hand, we are ready to prove the main norm estimates of this section.
\begin{proof}[Proof of Proposition \ref{prop:word_corner}]
Denote the lengths of the words $\mathbf{A}_1$ and $\mathbf{A}_2$ correspondingly $\ell_1$ and $\ell_2$ ($\ell_1+\ell_2=\ell-1$). If either $\ell_1=0$ or $\ell_2=0$ the estimate already follows from Corollary \ref{cor:important_estimate}.
Else, we rewrite:
\begin{multline}
P_n \mathbf{A}_1 H(f)\mathbf{A}_2 Q_n
\\=P_n \mathbf{A}_1 P_n H(f) \mathbf{A}_2 Q_n
+ P_n \mathbf{A}_1 Q_n H(f) Q_n \mathbf{A}_2 Q_n
+ P_n \mathbf{A}_1 Q_n H(f) P_n \mathbf{A}_2 Q_n.
\end{multline}
Then bounding summands individually, we have that
\begin{multline}
\|P_n \mathbf{A}_1 H(f)\mathbf{A}_2 Q_n\|_2
\leq
\|P_n \mathbf{A}_1\|_{op} \|P_n H(f) \mathbf{A}_2 Q_n\|_{2}
\\+ \|P_n \mathbf{A}_1 Q_n H(f) Q_n\|_2 \| \mathbf{A}_2 Q_n\|_{op}
\\+ \|P_n \mathbf{A}_1\|_{op} \|Q_n H(f) P_n\|_2 \|\mathbf{A}_2 Q_n\|_{op},
\\ \leq \kappa C^{\ell-1} 2(\ell_1+1) \sqrt{\log(2+\ell_1)} + \kappa C^{\ell-1} 2(\ell_2+1) \sqrt{\log(2+\ell_2)}+ \kappa C^{\ell-1} \sqrt{\log 2}
\\
\leq 3\kappa C^{\ell-1} \ell \sqrt{\log(1+\ell)},
\end{multline} 
where we used Corollary
\ref{cor:important_estimate} and Lemma \ref{lemma:hslemma} for the estimates on the Hilbert-Schmidt norm. The case $Q_n \mathbf{A}_1 H(f)\mathbf{A}_2 P_n$ can be reduced to the above by taking transposes.
\end{proof}
\begin{proof}[Proof of Proposition \ref{prop:word_trace}]
Expanding the commutator we have that
\begin{multline}\label{eqn:proof_lemma_word_trace_1}
\Tr P_n[ \mathbf{A}_1 H(f^{(1)}) \mathbf{A}_2 H(f^{(2)}) \mathbf{A}_3, T(b)]P_n
\\=\Tr P_n \mathbf{A}_1 H(f^{(1)}) \mathbf{A}_2 H(f^{(2)}) \mathbf{A}_3Q_nT(b)P_n
\\- \Tr P_n T(b) Q_n\mathbf{A}_1 H(f^{(1)}) \mathbf{A}_2 H(f^{(2)}) \mathbf{A}_3P_n.
\end{multline}
Note here that due to Corollary \ref{cor:bounded_completion} (see Remark \ref{rmk:important_estimate}),
$$\|P_nT(b)Q_n\|_{op}, \|Q_n T(b)P_n\|_{op}\leq \min(\|b\|_{\infty}, \kappa\pi, C)\leq C.$$
We rewrite the first summand in \eqref{eqn:proof_lemma_word_trace_1}, so that there are projections on the right of $H(f^{(1)})$.
\begin{multline}\label{eqn:trace_decomposition}
\Tr P_n \mathbf{A}_1 H(f^{(1)}) \mathbf{A}_2 H(f^{(2)}) \mathbf{A}_3Q_nT(b)P_n
\\
\begin{aligned}
&=\Tr P_n \mathbf{A}_1 H(f^{(1)})P_n \mathbf{A}_2 H(f^{(2)}) \mathbf{A}_3Q_nT(b)P_n
\\ &+\Tr P_n \mathbf{A}_1 H(f^{(1)})Q_n \mathbf{A}_2 H(f^{(2)}) \mathbf{A}_3Q_nT(b)P_n.
\end{aligned}
\end{multline}
Now, in \eqref{eqn:trace_decomposition}, the second term readily decomposes into a product of two operators with uniformly bounded Hilbert-Schmidt norm, see \eqref{eqn:4terms}. We use the cyclicity of trace (since the operators we move are trace-class) to also rewrite the first term as a trace of a trace-class operator:
\begin{multline}\label{eqn:non-trace-to-trace}
\Tr P_n \mathbf{A}_1 H(f^{(1)}) P_n
\mathbf{A}_2 H(f^{(2)}) 
\mathbf{A}_3Q_n(Q_nT(b)P_n)
\\= \Tr (Q_nT(b)P_n) 
P_n \mathbf{A}_1
H(f^{(1)})P_n 
\mathbf{A}_2 H(f^{(2)}) \mathbf{A}_3Q_n.
\end{multline}
Thus we replace the first term on the right-hand side in \eqref{eqn:trace_decomposition} with the right-hand side of \eqref{eqn:non-trace-to-trace}, and bound the trace on the left-hand side with the trace-norm on the right-hand side in this rewriting. We see that:
\begin{multline}\label{eqn:4terms}
|\Tr P_n \mathbf{A}_1 H(f^{(1)}) \mathbf{A}_2 H(f^{(2)}) \mathbf{A}_3Q_nT(b)P_n|
\\
\begin{aligned}
\leq & \|Q_nT(b)P_n P_n \mathbf{A}_1
H(f^{(1)})P_n\|_2
\|P_n \mathbf{A}_2 H(f^{(2)}) \mathbf{A}_3Q_n\|_2
\\&+ \|P_n \mathbf{A}_1 
H(f^{(1)})Q_n\|_2
\|Q_n \mathbf{A}_2 H(f^{(2)})
\mathbf{A}_3Q_nT(b)P_n\|_2.
\end{aligned}
\end{multline}
To estimate the Hilbert-Schmidt norms we use Corollary \ref{cor:important_estimate} and Proposition \ref{prop:word_corner}.
We show how to do this for the first term. Denoting the word length of 
$T(b)\mathbf{A}_1H(f^{(1)})$, $\ell_1$, and the word length of 
$\mathbf{A}_2 H(f^{(2)}) \mathbf{A}_3$, $\ell_2$, so that $\ell_1+\ell_2=\ell$, applying, correspondingly Corollary \ref{cor:important_estimate} and Proposition \ref{prop:word_corner},
 we obtain:
 \begin{multline}
 \|Q_nT(b)P_n P_n \mathbf{A}_1
 H(f^{(1)})P_n\|_2
 \|P_n \mathbf{A}_2 H(f^{(2)}) \mathbf{A}_3Q_n\|_2
 \\\leq 2 \kappa C^{\ell_1-1}\ell_1\sqrt{\log(1+\ell_1)}
 \times
 3\kappa C^{\ell_2-1}\ell_2\sqrt{\log(1+\ell_2)}
 \\\leq \frac{3}{2}\kappa^2 C^{\ell-2}\ell^2 \log\left(1+\frac{\ell}{2}\right).
 \end{multline}
Via an identical estimate,  we bound the other term in \eqref{eqn:4terms} and see that
\begin{multline}
|\Tr P_n \mathbf{A}_1 H(f^{(1)}) \mathbf{A}_2 H(f^{(2)}) \mathbf{A}_3Q_nT(b)P_n|
 \leq
2 \left(\frac{3}{2}\kappa^2 C^{\ell-2} \ell^2 \log(1+\ell)\right)
\\= {3} \kappa^2C^{\ell-2} \ell^2 \log(1+\ell).
\end{multline}
Analogously we see that the second summand in \eqref{eqn:proof_lemma_word_trace_1} satisfies
\begin{equation}
| \Tr P_n T(b) Q_n\mathbf{A}_1 H(f^{(1)}) \mathbf{A}_2 H(f^{(2)}) \mathbf{A}_3P_n|
\leq
{3} \kappa^2 C^{\ell-2} \ell^2 \log(1+\ell).
\end{equation}
We sum the two bounds to deduce the result.
\end{proof}

\subsection{Proof of Theorem 
\ref{thm:bound_on_easy_nested}}

\begin{proof}[Proof of Theorem \ref{thm:bound_on_easy_nested}]
To estimate the operator norm of $\mathbf{K}_m$ we expand it into $2^{m-1}$ words and each word in the expansion has an operator norm bounded by $C^m$, so that we get a bound $2^{m-1} C^m$.
Similarly, summing the Hilbert-Schmidt norm estimates coming from Proposition \ref{prop:word_corner} of each individual term appearing in Proposition \ref{prop:structural_1} we obtain (2). In the same way,
after rewriting
$\mathbf{K}_m=[T(f^{(m)}), \mathbf{K}_{m-1}]$,
and applying Proposition \ref{prop:structural_1} to
$\mathbf{K}_{m-1}$
we obtain (3) via Proposition \ref{prop:word_trace}.
\end{proof}

\section{Bounds on nested commutators for unbounded symbols}\label{sec:unbounded_bounds}
In this section we prove bounds on the nested commutators
$$
    \mathbf{K}_m=[T(f^{(m)}),\ldots, T(f^{(2)}), T(f^{(1)})],
$$
but we no longer assume boundedness for the symbols. In this setting we prove the following result.

\begin{theorem}\label{thm:bound_on_difficult_nested}
Assume that for $m\geq 2$ the symbols $f^{(1)},f^{(2)},\ldots, f^{(m)}$ satisfy \eqref{eq:condition_on_fk2} for some common $\kappa>0$ and
are admissible, as in Definition \ref{def:nehari_complete} with underlying constants $C, \kappa^*>0$. 
Then
\begin{enumerate}
\item $\|\mathbf{K}_m\|_{op}\leq 
C^m 2^{m-1}(m-1)!,
$
\item $\|P_n\mathbf{K}_mQ_n\|_{2}, \|Q_n\mathbf{K}_mP_n\|_{2}\leq  3\kappa^* (2C)^{m-1} m! \sqrt{\log(1+m)}$,
\item $\left|
\Tr  P_n \mathbf{K}_m P_n 
\right|
\leq 
6 \kappa^{*2}  (2C)^{m-2} m^2\log(1+m)(m-2)!,
\quad\text{for}\quad m\geq 3.$
\end{enumerate}
\end{theorem}

In fact, similarly to Theorem \ref{thm:bound_on_easy_nested}, we prove this result by expanding $\mathbf{K}_m$ in polynomials of length $\ell$ of Hankel and Toeplitz operators. A remarkable fact is that the symbols of these operators can be assumed to be bounded. 

\subsection{Word structure}\label{sec:word_structure}
We will once again begin with a monomial expansion of the given $m$-nested commutator. This time around, obtaining a suitable expansion requires extra work. 
With the assumed notation in Definition \ref{def:nehari_complete}, for an \textit{admissible} symbol $f$ we denote the bounded symbols
$$f_*=f_++g_-\quad \text{and}\quad f_{**}=g_++f_-,$$
so that now $f_*$ and $f_{**}$ satisfy \eqref{eq:condition_on_fk}, but also $\|f_*\|, \|f_{**}\|<C$. For the way we make use of them, we will refer to these symbols as \textit{replacement symbols}.

\begin{proposition}\label{prop:structure_of_nehari_family}
For $m\geq 2$, suppose $f^{(1)}, f^{(2)}, \ldots, f^{(m)}
$ satisfy \eqref{eq:condition_on_fk2} for a common $\kappa>0$ and furthermore assume they are admissible (recall Definition \ref{def:nehari_complete}) with the same constants, $\kappa^*, C>0$.
Then $\mathbf{K}_m$
is expressed as a sum of monomials(words) of degree $m$ of Toeplitz and Hankel operators of the symbols
$$f^{(k)}_*, f^{(k)}_{**}, \widetilde{f^{(k)}_*}, \widetilde{f^{(k)}_{**}},\quad \text{for}\quad k=1,\ldots, m.$$ 
Each word $\mathbf{W}_m$ in this expansion contains all indices and
the number of Hankel operators in each word is even ($\geq 2$).
The number of words in the expansion, $N_m$, is at most $2^{m-1} (m-1)!$
\end{proposition}
This Proposition gives access to norm bounds of $\mathbf{K}_m$ via the norm bounds on words of bounded symbols established in Propositions \ref{prop:word_corner} and \ref{prop:word_trace}. Unfortunately, at the cost of a $(m-1)!$ factor. Before going on to prove Proposition \ref{prop:structure_of_nehari_family} we proceed to prove Theorem \ref{thm:bound_on_difficult_nested}.

\subsection{Proof of Theorem \ref{thm:bound_on_difficult_nested}}
\begin{proof}[Proof of Theorem \ref{thm:bound_on_difficult_nested}]
The fundament of the proof is the word expansion in Proposition \ref{prop:structure_of_nehari_family}. As in the proof of Theorem \ref{thm:bound_on_easy_nested}, we estimate the desired norms and the trace individually for each word and then perform a crude estimate, summing all bounds.

Since the replacement symbols, $\|f_*^{(k)}\|_{\infty}, \|f_{**}^{(k)}\|_{\infty}\leq C$, due to Theorems \ref{thm:nehari}, \ref{thm:brown_halmos}, $\|\mathbf{W}_m\|_{op}\leq C^m$.

Due to Proposition \ref{prop:word_corner},
$$\|P_n\mathbf{W}_mQ_n\|_2, \|Q_n\mathbf{W}_mP_n\|_2\leq 3\kappa^* C^{m-1} m \sqrt{\log(1+m)},$$
and Proposition \ref{prop:word_trace} tells us that
$$\left| \Tr  P_n[\mathbf{W}_{m-1}, T(f^{(m)})]P_n \right|
\leq
6 \kappa^{*2}  C^{m-2} m^2 \log(1+m).$$
Thus the theorem reduces to counting the number of words  $\mathbf{W}_{m}$ in $\mathbf{K}_m$ for (1) and (2),
the number of words  $\mathbf{W}_{m-1}$ in $\mathbf{K}_{m-1}$ for (3) which are bounded correspondingly by $2^{m-1}(m-1)!$ and $2^{m-2}(m-2)!$  by Proposition \ref{prop:structure_of_nehari_family}. We conclude the result.
\end{proof}
\subsection{Replacement lemma and a twisted commutator}
The mechanism we utilize to replace unbounded symbols with bounded ones are \textit{twisted commutators}.
\begin{definition}
For operators $A=(A_{ij})_{i,j\in\mathbb{N}},  B=(B_{ij})_{i,j\in\mathbb{N}}$ with a matrix representation, we define the \textbf{twisted commutator}, $[\cdot, \cdot]_S$ by
\begin{equation}
[A,B]_S:=AB-B^TA; \quad ([A,B]_S)_{ij}=\sum_{k=1}^{\infty}A_{ik}B_{kj}- B_{ki}A_{kj}.
\end{equation}
\end{definition}

\begin{lemma}[Replacement lemma for twisted commutators]\label{lem:replacement}
Let $f, b\in L^2(\mathbb{T})$ with corresponding Fourier coefficients $(f_n)_{n\in\mathbb{Z}}$ and $(b_n)_{n\in\mathbb{Z}}$. Then 
$[H(b), T(f)]_S$ depends only on the Fourier coefficients $f_{-n}, b_n$ for $n>0$. In particular for any $f_{**}, b_*\in L_2(\mathbb{T})$ with Fourier coefficients 
$(f_{**n})_{n\in\mathbb{Z}}$, $(b_{*n})_{n\in\mathbb{Z}}$ such that for $n>0$, $f_{**-n}=f_{-n}$ and $b_{*n}=b_n$, we have that
$$[H(b), T(f)]_S = [H(b_*), T(f_{**})]_S.$$ 
\end{lemma}
\begin{proof}
For a general matrix $K=(K_{ij})_{i,j\in\mathbb{N}}$, we have that
\begin{multline}
([K, T(f)]_S)_{ij}
= \sum_{k=1}^{\infty}\left(K_{ik}f_{k-j}- f_{k-i}K_{kj}\right)
\\=\sum_{k=1}^{j-1} f_{k-j} K_{ik}
+\sum_{k=j}^{\infty} f_{k-j} K_{ik}
- \sum_{k=1}^{i-1} f_{k-i}K_{kj} - \sum_{k=i}^{\infty} f_{k-i}K_{kj}
\\=\sum_{s=0}^{\infty} f_s(K_{i, j+s}- K_{i+s, j}) + \sum_{s=1}^{j-1} f_{-s}K_{i, j-s}- \sum_{s=1}^{i-1} f_{-s}K_{i-s, j}.
\end{multline}
Thus for $K=H(b)$, we have
\begin{equation}\label{eqn:twisted_hankel}
([H(b), T(f)]_S)_{ij}= (-1)^{\mathbbm{1}_{i>j}}    \sum_{s=\min(i,j)}^{\max(i,j)-1} f_{-s}b_{i+j-s-1},
\end{equation} 
and we see the last expression depends only on the negative Fourier coefficients of $f$ (and positive Fourier coefficients of $b$).
\end{proof}
\begin{remark}
The expression for the twisted commutator \eqref{eqn:twisted_hankel} captures a remarkable cancellation, that is not apparent for $[H(b), T(f)]$ or for 
\\$[H(b)H(c),T(f)]$ from entry-wise computations to begin with.
\end{remark}
It is illustrative to see how one may express the $3-$nested commutator $\mathbf{K}_3$ to include the twisted commutator.
\begin{example}
\begin{multline}
[T(f_+),[T(f_+),T(f_-)]]=-[T(f_+), H(f_+)H(\widetilde{f_-})]
\\=H(f_+)[H(\widetilde{f_-}),T(f_+)]_S+[H(f_+), T(\widetilde{f_+})]_S H(\widetilde{f_-})
\\=[H(f_+), T(\widetilde{f_+})]_S H(\widetilde{f_-}).
\end{multline}
\end{example}

In fact, we see that:
\begin{lemma}\label{lem:bounded_twisted}
Suppose $f:\mathbb{T}\rightarrow\mathbb{C}$ satisfies \eqref{eq:condition_on_fk2}. Then
$$|([H(f_+), T(\widetilde{f_+})]_S)_{ij}|\leq \frac{2\kappa^2|\log(\frac{j-1/2}{i-1/2})|}{i+j-1}.$$
\end{lemma}
\begin{proof}
We use \eqref{eqn:twisted_hankel}. Assume $j>i$. We have that:
\begin{multline}
|([H(f_+), T(\widetilde{f_+})]_{S})_{ij}|
  \leq \sum_{s=i}^{j-1}\frac{\kappa}{s}\frac{\kappa}{i+j-s-1}
  =\frac{2\kappa^2}{i+j-1}\sum_{s=i}^{j-1}\frac{1}{s}
  \\\leq \frac{2\kappa^2}{i+j-1}\int_{i-\frac{1}{2}}^{j-\frac{1}{2}}\frac{dx}{x}
  =\frac
  {2\kappa^2\log(\frac{j-1/2}{i-1/2})}
  {i+j-1}.
\end{multline}
The case $j<i$ is checked similarly and for $i=j$ we get $0$.
\end{proof}
From Lemma \ref{lem:bounded_twisted} (with some more work) it is possible to deduce the boundedness of $[H(f_+), T(\widetilde{f_+})]_{S}$ as well as the fact that the Hilbert-Schmidt norm of its corners is bounded. We take a different approach to prove this for $m-$nested commutators utilizing Lemma \ref{lem:replacement}. 

\subsection{Proof of Proposition \ref{prop:structure_of_nehari_family}}

\begin{proof}[Proof of Proposition \ref{prop:structure_of_nehari_family}]
We proceed by induction on $m\geq 2$. 
For $m=2$,
$$\mathbf{K}_2= H(f^{(1)})H(\widetilde{f^{(2)}})- H(f^{(2)})H(\widetilde{f^{(1)}})=H(f^{(1)}_{*})H(\widetilde{f^{(2)}_{**}})- H(f^{(2)}_{*})H(\widetilde{f^{(1)}_{**}})$$
satisfies the induction hypothesis.
\medskip

For $m>2$, under the induction hypothesis
$\mathbf{K}_{m-1}$ is a linear sum of $N_{m-1}\leq 2^{m-2}(m-2)!$ monomials of the desired form.

We prove that applying a commutator to a monomial of degree $m-1$ produces at most $2(m-1)$ new monomials satisfying the induction hypothesis. 
Assume $b^{(1)}, b^{(2)}, \ldots, b^{(m-1)}$ are replaced symbols and denote $B(b^{(k)})$ to be either $H(b^{(k)})$ or $T(b^{(k)})$, depending on the monomial. We wish to show that
$$[B(b^{(1)})B(b^{(2)})\ldots B(b^{(m-1)}), T(f^{(m)})]$$
is of the desired form. We utilize an identity similar to the generalized Leibniz identity for commutators:
$$
[A_1A_2\ldots A_k, B]=\sum_{p=1}^{k}A_1\ldots A_{k-p}[A_{k-p+1},B]A_{k-p+2}\ldots A_k;
$$
 Only instead of taking the usual commutator every time, an even number of times, whenever $B(b^{(p)})$ is a Hankel operator, we apply a symmetric commutator instead and in that case we also need to transpose $T(f^{(m)})$, that is replace it by $T(f^{(m)})^T=T(\widetilde{f^{(m)}})$ for the terms to follow.
 For instance:
\begin{multline}\label{eqn:twisted_example}
[H(b^{(1)})T(b^{(2)})T(b^{(3)})H(b^{(4)})T(b^{(5)}), T(f)]
\\=H(b^{(1)})T(b^{(2)})T(b^{(3)})H(b^{(4)})T(b^{(5)})
T(f)
\\-
T(f)
H(b^{(1)})T(b^{(2)})T(b^{(3)})H(b^{(4)})T(b^{(5)})
\\\begin{aligned}
&= H(b^{(1)})T(b^{(2)})T(b^{(3)})H(b^{(4)})[T(b^{(5)}), T(f)]
\\&+ H(b^{(1)})T(b^{(2)})T(b^{(3)})[H(b^{(4)}), T(f)]_S T(b^{(5)})
\\&+H(b^{(1)})T(b^{(2)})[T(b^{(3)}), T(\widetilde{f})]H(b^{(4)})T(b^{(5)})
\\&+H(b^{(1)})[T(b^{(2)}), T(\widetilde{f})]T(b^{(3)})H(b^{(4)})T(b^{(5)})
\\&+[H(b^{(1)}), T(\widetilde{f})]_ST(b^{(2)})T(b^{(3)})H(b^{(4)})T(b^{(5)}).
\end{aligned}
\end{multline}

See in the example, that by taking commutators we move $T(f)$ in the first summand on the second line of \eqref{eqn:twisted_example} one symbol to the left. Only a twisted commutator flips it to $T(\widetilde{f})$ and so, to eventually cancel the second summand, we need to perform an even number of such flips. Since there is an even number of Hankel operators in each monomial by the inductive hypothesis, we can always represent the commutator as above.
Now for $k\geq 1$,
 $$[T(b^{(k)}), T(f)]=H(f)H(\widetilde{b^{(k)}})- H(b^{(k)})H(\widetilde{f})=H(f_*)H(\widetilde{b^{(k)}})-H(b^{(k)})H(\widetilde{{f}_{**}}),$$
 and the computation is similar for $[T(b^{(k)}), T(f)^T]=[T(b^{(k)}), T(\widetilde{f})]$.

Due to Lemma \ref{lem:replacement}, and the admissibility of the symbols (recall Definition \ref{def:nehari_complete}), we may in fact replace the symbol $f$ with its bounded replacement in each twisted commutator:
$$[H(b^{(k)}), T(f)]_S=[H(b^{(k)}),T(f_{**})]_S= H(b^{(k)})T(f_{**})-T(\widetilde{f_{**}})H(b^{(k)})$$
and similarly for $[H(b_k), T(\widetilde{f})]_S$.

The number of Hankels in a monomial increases by two or stays the same after the procedure. By the above observations all new monomials are of the form desired in the hypothesis and furthermore out of a single word of length $m$, we obtain at most $2m$ new monomials of length $m+1$.
The induction is complete.
\end{proof}

In fact in this proof admissibility was only needed to deduce that the (replacement) symbols appearing in the monomials still satisfy \eqref{eq:condition_on_fk2}. Relaxing that assumption, we can deduce the boundedness of $m-$nested commutators under the sole assumption of \eqref{eq:condition_on_fk2}.

\begin{corollary}\label{cor:nehari+hilbert}
Suppose for $m\geq 2$, $f^{(1)}, f^{(2)},\ldots, f^{(m)}$ satisfy \eqref{eq:condition_on_fk2} for the same $\kappa>0$.
Then the nested commutator $\mathbf{K}_m$ is bounded. 
In fact,
$$\|
[T(f^{(m)}),\ldots, T(f^{(2)}), T(f^{(1)})]
\|_{op}
\leq (\kappa\pi)^m 2^{m-1}(m-1)!.
$$
\end{corollary}
\begin{proof}
We apply the same replacement procedure as in the proof of Proposition \ref{prop:structure_of_nehari_family}, only this time the replacement symbols are chosen via Corollary \ref{cor:bounded_completion}. 
Thus each replacement symbol is bounded by $\kappa\pi$. We conclude the result.
\end{proof}

\section{Proof of main results}\label{sec:proofs}
\subsection{The main idea}

To avoid technical assumptions in the initial discussion we assume that for some $C>0$, $\|f_+\|_{\infty}, \|f_-\|_\infty< C$, and $f_0=0$. 
We briefly recap that we can rewrite the Toeplitz determinant:
\begin{multline}\label{eqn:z(t)expression}
\det T_n (e^{tf})
=e^{ntf_0}
\det(P_n e^{-tT(f_+)}e^{tT(f_-)}e^{tT(f_+)}e^{-tT(f_-)}P_n+Q_n)
\\= \det(P_n e^{Z(t)}P_n+Q_n).
\end{multline}
In Sections \ref{sec:norm_estimates}, \ref{sec:unbounded_bounds} we established bounds on the Taylor coefficients of $Z(t)$. We show how these lead us to our main results via an expansion for Fredholm determinants given in Theorem \ref{thm:generalthm}.
\subsection{Operator determinants}\label{sec:operator_determinants}

To study the right-hand side of \eqref{eqn:z(t)expression}, we use a result for determinants of finite sections of full operators on a separable Hilbert space $\mathcal{H}$. 
Provided that these operators have a logarithm $B:\mathcal{H}\rightarrow\mathcal{H}$, %
recalling the definitions of the orthogonal projections $P_n, Q_n,$ in \eqref{eqn:pnqn}, we have the following statement:
\begin{theorem}\label{thm:generalthm}
Assume that $B:\ell_2\rightarrow\ell_2$ is an operator defined by an infinite matrix such that
\begin{enumerate}
\item $B$ is bounded
\item $\|P_nBQ_n\|_2, \|Q_nBP_n\|_2<\infty$, for all $n$.
\end{enumerate}
Then {for}  $ \|B\|_{op} \leq \frac{1}{3}$, $\log\det(P_ne^{B}P_n+Q_n)$ satisfies
$$\left|\log\det(P_ne^{B}P_n+Q_n)-\Tr (P_nBP_n)\right|\leq A \|P_nBQ_n+ Q_nBP_n\|_2^2,$$
where we can take $A=e^2\sum_{m=0}^{\infty} (e/3)^m (m+2)^{3/2}$.
\end{theorem}

Theorem \ref{thm:generalthm} is closely related to the following result from \cite{BreuerDuits}.

\begin{theorem}\label{thm:Breuer-Duits}
Let $K$ be a self-adjoint projection operator (i.e. $K^2=K$ and $K^*=K$) on a separable Hilbert space. Then there exists a constant $A>0$ (independent of $K$ and $h$) such that for any bounded operator $h$ such that $hK$ and $Kh$ are of trace class we have 
\begin{equation}
\left|\log \det(1+(e^{th}-1)K)-t \Tr (hK)\right| \leq A|t|^2\|[h,K]\|_2^2, \quad\text{for $|t| \|h\|_{op}\leq\frac{1}{3}$.}
\end{equation}
Here, $[h,K]=hK-Kh$ stands for the commutator of $h$ and $K$. They have
$$A=e^2\sum_{m=0}^{\infty} (e/3)^m (m+2)^{3/2}.$$

\end{theorem}

\begin{proof}[Proof of Theorem \ref{thm:generalthm}]
Since $P_n$ is a self-adjoint projection operator in $\ell_2$, Theorem \ref{thm:Breuer-Duits} assures that 
$$\left|\log \det(1+(e^{tB}-1)P_n)-t \Tr (BP_n)\right| \leq A|t|^2\|[P_n,B]\|_2^2, \quad\text{for $|t|\leq\frac{1}{3\|B\|_{op}}$.}
$$
Then one can deduce the result from the observations that 
$$\Tr (BP_n)=\Tr (P_nBP_n),$$
that
\begin{equation}\label{eqn:determinantdeduction}
\det(1+(e^{tB}-1)P_n)=\det(1+P_n(e^{tB}-1)P_n)=\det(Q_n+P_ne^{tB}P_n),
\end{equation}
and that 
$$\|[P_n,B]\|_2=\|P_nBQ_n-Q_nBP_n\|_2= \|P_nBQ_n+Q_nBP_n\|_2.$$
\end{proof}

\subsection{Bounded symbols}
We first compute the asymptotic expansion of the determinant in a neighborhood of zero. 
\begin{theorem}\label{thm:determinant_expansion}
Suppose $f:\mathbb{T}\rightarrow\mathbb{C}$ satisfying \eqref{eq:condition_on_fk2}, $f_0=0$, and, moreover, there is $C>0$, $\|f_+\|_{\infty}, \|f_-\|_{\infty}< C$. Then
\begin{equation}
        C_2^{(n)}(f)= \sum_{k=1}^\infty \min (k,n) f_kf_{-k},
    \end{equation}
and for $|t|\leq 
\frac{\sqrt{13}-1}{6} 
\frac{1}{8Ce}$,
$\log\det T_n(e^{tf})$ is analytic in $t$ and
\begin{equation}
\log\det T_n(e^{tf})= t^2\sum_{k=1}^\infty \min (k,n) f_kf_{-k}+ \mathcal{O}(|t|^3),\quad\text{as}\quad n\rightarrow\infty,
\end{equation}
uniformly in $n$ and, moreover,
$$|
\Psi_n(t)-t^2
\sum_{k=1}^{\infty} 
\min(k,n)
f_kf_{-k} 
|\leq \left(\frac{\kappa}{C}\right)^2
\left(1+\frac{A}{2}
\right),$$
where $A>0$ is the constant in Theorem \ref{thm:generalthm}.
\end{theorem}
\begin{proof}
Since $T(f_+), T(f_-)$ are bounded, all identities performed in \eqref{eqn:z(t)expression} are valid in a neighborhood of zero. 
In the expansion
\begin{equation}\label{eqn:z(t)ksum}
Z(t)=\sum_{m=2}^{\infty}t^m\mathbf{K}_m^{(sum)},
\end{equation}
by Corollary \ref{cor:k-commutators} (the case $k=4$), $\mathbf{K}_m^{(sum)}$ is a linear expression of $m$-nested commutators, $\mathbf{K}_m$, with the sum of the absolute values of their coefficients bounded by $\frac{(4e)^m} 
{\sqrt{2\pi m^3}}$. In this expression, we bound the norm of each summand individually via Theorem \ref{thm:bound_on_easy_nested} and apply the triangle inequality to deduce bounds on the norm of $\mathbf{K}_m^{(sum)}$:
$\|\mathbf{K}_m^{(sum)}\|\leq \frac{(4e)^m}
{\sqrt{2\pi m^3}} \|\mathbf{K}_m\|$.
Explicitly, we have:
\begin{enumerate}
\item 
$
\|\mathbf{K}_m^{(sum)}\|_{op}\leq \frac{(8Ce)^m}{2\sqrt{2\pi m^3}},
$
\item 
$
\|P_n\mathbf{K}_m^{(sum)}Q_n\|_{2}, \|Q_n\mathbf{K}_m^{(sum)}P_n\|_{2}\leq \frac{(8Ce)^m}{\sqrt{2\pi m}} \frac{3}{2}\frac{\kappa}{C} \sqrt{\log(1+m)},$
\item 
$|\Tr  P_n \mathbf{K}_m^{(sum)} P_n| 
\leq \frac{3}{2\sqrt{2\pi}} 
\left(\frac{\kappa}{C}\right)^2 
(8eC)^m
\sqrt{m}
\log(1+m)
\quad\text{for}\quad m\geq 3.$
\end{enumerate}
We check that for $|t|\leq \frac{\sqrt{13}-1}{6}\frac{1}{8Ce}$,
$$\|Z(t)\|_{op}\leq\sum_{m=2}^{\infty}
|t|^m
\|\mathbf{K}_m^{(sum)}\|_{op}\leq 1/6,\quad \|P_nZ(t)Q_n\|_2 , \|Q_nZ(t)P_n\|_2=\mathcal{O}(t^2)\leq \frac{\kappa}{2C},$$
$$\text{and}\quad\left|\Tr  P_nZ(t)P_n
- t^2\Tr  P_n\mathbf{K}_2^{(sum)}P_n\right|
=\mathcal{O}(|t|^3)
\leq \left(\frac{\kappa}{C}\right)^2.$$
Hence, for $|t|\leq \frac{\sqrt{13}-1}{48Ce}$,
$Z(t)$ satisfies the hypothesis of Theorem \ref{thm:generalthm} and, since here $\Tr P_n\mathbf{K}_2^{(sum)}P_n
=\Tr P_nH(f_+)H(\widetilde{f_-})P_n
=\sum_{k=1}^{\infty} 
\min(k,n)f_kf_{-k}
$,
\begin{equation}
\Psi_n(t)=t^2
\sum_{k=1}^{\infty} 
\min(k,n)f_kf_{-k}
+
\mathcal{O}(t^3).
\end{equation}
Explicitly, we furthermore have
$$
|
\Psi_n(t)-t^2
\sum_{k=1}^{\infty} 
\min(k,n)
f_kf_{-k} 
|
 \leq \left(\frac{\kappa}{C}\right)^2\left(1+\frac{A}{2}
 \right),$$
where $A>0$ is the constant in Theorem \ref{thm:generalthm}.
\end{proof}

We are now ready to deduce bounds for the cumulants.

\begin{corollary}\label{cor:bounds_in_bounded}
Suppose we are in the setting of Theorem \ref{thm:determinant_expansion}. Then for $0<R\leq\frac{\sqrt{13}-1}{48Ce},$ $A>0$ as in Theorem \ref{thm:generalthm},
\begin{equation}
|C_m^{(n)}(f)|\leq 
\left(\frac{\kappa}{C}\right)^2\left(1+\frac{A}{2}
\right)
R^{-m},
\end{equation}
for $m\geq 3.$
\end{corollary}
\begin{remark}
We do not expect the constant in the inequality to be optimal. 
\end{remark}
\begin{proof}
In Theorem \ref{thm:determinant_expansion}, the form of the remainder term implies that in fact $|C_k^{(n)}(f)|=\mathcal{O}(1)$, since the radius of convergence is uniformly bounded. For 
$|R|\leq\frac{\sqrt{13}-1}{48Ce},$ 
\begin{multline}
C_{m}^{(n)}(f)
= \frac{1}{2\pi i}\oint_{|t|=R}\log \det T_n(e^{tf})\frac{dt}{t^{m+1}}
\\=  \frac{1}{2\pi i}\oint_{|t|=R}\log \det P_n e^{Z(t)}P_n\frac{dt}{t^{m+1}}
\\= \frac{1}{2\pi i}\oint_{|t|=R} 
\left(\Tr  P_nZ(t)P_n +\mathcal{O}(1)\right)\frac{dt}{t^{m+1}}
\\= \frac{1}{2\pi i}\oint_{|t|=R} \left(t^2\Tr P_n \mathbf{K}_2 P_n +\mathcal{O}(1)\right)\frac{dt}{t^{m+1}},
\end{multline}
and we note that the remainder is analytic in $t$ and has a zero of order at least $3$ at $t=0$.
Since $Z(t)$ has Taylor coefficients of degree greater than $2$, $C_0^{(n)}(f)=C_{1}^{(n)}(f)=0$, and
$$
C_{2}^{(n)}(f)=\Tr P_n \mathbf{K}_2P_n=\sum_{k=1}^{\infty}\min(k,n) f_k f_{-k},
$$
and using the explicit bound for the $\mathcal{O}(1)$ term, we obtain
$$|C_{m}^{(n)}(f)|
\leq \left(\frac{\kappa}{C}\right)^2\left(1+\frac{A}{2}
\right) R^{-m}.
$$
\end{proof}

\subsection{Unbounded symbols}\label{sec:unbounded_cumulants}
If either $f_+$ or $f_-$ is unbounded, it is not clear whether the formal expansion that comes from Theorem \ref{thm:k-commutators} and Corollary \ref{cor:k-commutators} makes sense analytically. In this section we make sure that it does in an appropriate sense. For that purpose we first ensure that indeed the cumulant generating function is analytic in a neighborhood of zero. Then we demonstrate, via a regularization that indeed the provided expansion is the desired one.  
To show the analyticity of the moment generating function, we show that the Fourier coefficients of $\exp(tf)$ are analytic in $t$. 
But for an unbounded test function $f$, this is not apriori true. To that end, we provide sufficient conditions for this to be true.

\begin{lemma}\label{lemma:barely_coefficients}
Suppose $f$ satisfies \eqref{eq:condition_on_fk}. Then, for sufficiently small $t$, we have $\exp({tf(z)})\in L_2(\mathbb T)$ and  the Fourier coefficients of $\exp({tf(z)})$ are analytic functions of $t$. Moreover, 
\begin{equation}\label{eqn:analytic_fc}
\left( \exp{tf(z)} \right)_k=\sum_{m=0}^{\infty}\frac{t^m}{m!}\sum_{k_1+\ldots+k_m=k} f_{k_1}\ldots f_{k_m}.
\end{equation}
\end{lemma}
\begin{proof}
Since \eqref{eq:condition_on_fk} holds there is a $\kappa>0$ such that $|f_k|\leq \frac{\kappa}{|k|}$.
For $r<1$, we define $f_r(z)=\sum_{k=1}^{\infty} (f_k(zr)^k+f_{-k}(z^{-1}r)^k),$ a bounded function. The Fourier coefficients of $\exp(tf_r)$ are given by 
$$\left( \exp{tf_r(z)} \right)_k=\sum_{m=0}^{\infty}\frac{t^m}{m!}\sum_{k_1+\ldots+k_m=k} r^{|k_1|+\ldots +|k_m|} f_{k_1}\ldots f_{k_m}.$$ 
We use a comparison with $g(z)=-\frac{1}{4+\delta}\left(\log(1-z)+\log(1-z^{-1})\right),$
observing that for $t$ sufficiently small (explicitly $|t|<\frac{1}{\kappa(4+\delta)}$ for some $\delta>0$), $|tf_m|\leq g_m=\frac{1}{(4+\delta)|m|}.$ 
Setting $h_k(t):=\lim_{r\uparrow 1}\left( \exp{tf_r(z)} \right)_k$, it follows that
$$
|h_k(t)|:=|\lim_{r\uparrow 1}\left( \exp{tf_r(z)} \right)_k|\leq \exp(g(z))_k<\infty,
$$
and due to this bound $h_k(t)$ is analytic in $t$ (for $t$ sufficiently small) and its Taylor series is given by the right-hand side of \eqref{eqn:analytic_fc}.
In fact $$\sum_k |(\exp(g))_k|^2=\int_{\mathbb{T}}\left|\exp(2g(z))\right|dz=\int_{\mathbb{T}}|1-z|^{-\frac{1}{1+\delta/4}}dz<\infty.$$
Observe now that, by Parseval's identity, since $|\left( \exp{tf_r(z)} \right)_k|<(\exp(g))_k$  , as $r\uparrow 1$, $\exp(tf_r(z))\xrightarrow{L^2(\mathbb{T})} h(t)(z)=\sum_{k}h_k(t)z^k$. Hence, along some subsequence $r_n$, $\exp(t f_{r_n})\rightarrow h(t)(z)$ a.e.. On the other hand $\exp(tf_r(z))\rightarrow \exp(tf(z))$ a.e. and so $\exp(tf(z))=h(t)(z)$ a.e.. 
\end{proof}

\begin{proposition}\label{prop:analyticity_unbounded}
Suppose $f$ satisfies \eqref{eq:condition_on_fk}. Then for all $n\in\mathbb{N}$, $\det T_n(e^{tf})$ and $\log\det T_n(e^{tf})$ are well-defined analytic functions in a neighborhood of $t=0$. Moreover, there is an absolute constant $R_*(n)$, which bounds the radius of convergence from below, and it depends only on $n$ and the upper bound on the Fourier coefficients.
\end{proposition}

\begin{proof}
This result follows immediately from Lemma \ref{lemma:barely_coefficients}. 
Namely, the determinant is a polynomial and the coefficients are analytic for $t$ sufficiently small. Via an absolute estimate we can check further that for $t$ sufficiently small $T_n$ is sufficiently close to the identity so that we can take a logarithm of the determinant.
\end{proof}

    \begin{theorem}\label{thm:unbounded_cumulants}
Suppose the function $f$ satisfies \eqref{eq:condition_on_fk2} and is admissible with constants $\kappa^*, C>0$ (recall Definition \ref{def:nehari_complete}). Then, for $n \in \mathbb N$,
 \begin{equation}
        C_2^{(n)}(f)= \sum_{k=1}^\infty \min (k,n) f_kf_{-k},
    \end{equation}
    and for $m\geq 3$,
   \begin{equation}\label{eq:unbounded_case_cumulant_bound}
        \left|C_m^{(n)}(f)\right|\leq 
        \left(\frac{\kappa^*}{C}\right)^2 \left(\frac{1}{32}A +\frac{7}{24}
\right)\left(24Cm\right)^{m},
    \end{equation}
    where $A=e^2\sum_{m=0}^{\infty} (e/3)^m (m+2)^{3/2}$ is as in Theorem \ref{thm:generalthm}.
    In particular, for $m\geq 3,$ we have $ C_m^{(n)}(f)=\mathcal O(1)$  as $n \to \infty$.
\end{theorem}
\begin{remark}
    We do not expect the bound \eqref{eq:unbounded_case_cumulant_bound} to be optimal. At various places, we have made rather crude estimates and it would be very interesting to see if, under this generality, the bound could be optimized to the form $\left|C_m^{(n)}(f)\right|\leq a c^m$ for some $c>0$, depending on $f$ but not $n$, as we proved for the bounded case.
\end{remark}

\begin{proof}
For $r<1$, we define the Poisson regularization of $f$ via $f_r(z)=f_+(rz)+ f_-(z/r),$ as in the proof of Lemma \ref{lemma:barely_coefficients}. Observe that under our assumptions $f_r, f_{r+}, f_{r-}$ are all bounded: since $|f_k|\leq \frac{\kappa}{|k|}$, $|f_{rk}|\leq \frac{r^{|k|}\kappa}{|k|}$, it follows that $$\|f_r\|_{\infty},\|f_+(rz)\|_{\infty},\|f_-({z/r})\|_{\infty}<-2\kappa\log(1-r).$$
Moreover by the maximum principle they are also admissible with the same underlying constant $\kappa^*$.

By Proposition \ref{prop:analyticity_unbounded} (assuming its notation), for $R<R_*(n)$, we have
\begin{multline}\label{eqn:approximation_C_m}
C_{m}^{(n)}(f)
= \frac{1}{2\pi i}\oint_{|t|=R}\log \det T_n(e^{tf})\frac{dt}{t^{m+1}}
\\=\frac{1}{2\pi i}\oint_{|t|=R}\log \det 
\lim_{r\uparrow 1} T_n(e^{tf_r})\frac{dt}{t^{m+1}}
\\=\lim_{r\uparrow 1}\frac{1}{2\pi i}\oint_{|t|=R}\log \det T_n(e^{tf_r})\frac{dt}{t^{m+1}}
=\lim_{r\uparrow 1} C_m^{n}(f_r).
\end{multline}
By Theorem \ref{thm:determinant_expansion}, since $f_r$ is bounded, there is a $R_{**}(r)>0$ such that for $R<\min(R_*(n), R_{**}(r))$,
\begin{multline}\label{eqn:cutoff_perform}
C_{m}^{(n)}(f_r)
= \frac{1}{2\pi i}\oint_{|t|=R}\log \det T_n(e^{tf_r})\frac{dt}{t^{m+1}}
\\=  \frac{1}{2\pi i}\oint_{|t|=R}\log \det P_n e^{Z_r(t)}P_n\frac{dt}{t^{m+1}}
\\=  \frac{1}{2\pi i}\oint_{|t|=R}\log \det P_n e^{Z_r^{(m)}(t)}P_n\frac{dt}{t^{m+1}}
,
\end{multline}
where
\begin{equation}
Z_r(t)=\sum_{p=2}^{\infty}t^p\mathbf{K}_p^{(sum)}(f_r), \text{ and for } m\geq 2,\quad Z_r^{(m)}(t)=\sum_{p=2}^{m}t^p\mathbf{K}_p^{(sum)}(f_r).
\end{equation}
The last equality in \eqref{eqn:cutoff_perform} is justified by the fact that powers of $t$ greater than $m$ do not affect the integral.
By Theorem \ref{thm:bound_on_difficult_nested}, the norms $\|\mathbf{K}_p^{(sum)}(f_r)\|_{op}$, $\|P_n\mathbf{K}_p^{(sum)}(f_r)Q_n\|_2,$ $\|Q_n\mathbf{K}_p^{(sum)}(f_r)P_n\|_2$ are bounded independently of $n$ and $r$, and so there is constant $R_{in}(m)$, independent of $n$ and $r$
(but depending on $m$) 
for which 
$$\sum_{p=2}^{m}R_{in}^p\|\mathbf{K}_p^{(sum)}(f_r)\|_{op}\leq 1/3\quad\text{for}\quad 0<r<1.$$
By Theorem \ref{thm:generalthm}, for $|t|\leq R_{in}$
\begin{multline}
\log \det P_n e^{Z_r^{(m)}(t)}P_n
\\=  \Tr P_n Z_r^{(m)}(t)P_n+ \mathcal{O}\left( \|P_n Z_r^{(m)}(t) Q_n+Q_nZ_r^{(m)}(t) P_n\|_2^2\right)
\\=t^2 \Tr P_n K_2^{(sum)}(f_r)P_n +\mathcal{O}(|t|^3)=t^2 \sum_{k=1}^{\infty}\min(n, k) r^{2k} f_k f_{-k}+\mathcal{O}(|t|^3),
\end{multline}
and the implicit bound for the remainder is independent of $n$ and $r$. The remainder term is analytic in $t$ for $R\leq R_{in}$ and has a zero of order at least $3$ at zero. Hence 
\begin{multline}
C_{m}^{(n)}(f_r)
=\frac{1}{2\pi i}\oint_{|t|=R}\log \det P_n e^{Z_r^{(m)}(t)}P_n\frac{dt}{t^{m+1}}
\\
=\frac{1}{2\pi i}\oint_{|t|=R_{in}}\log \det P_n e^{Z_r^{(m)}(t)}P_n\frac{dt}{t^{m+1}}
=\mathcal{O}(R_{in}^{-m}).
\end{multline}
Thus, from \eqref{eqn:approximation_C_m} it follows that 
$$C_2^{(n)}(f)=\lim_{r\uparrow 1} \sum_{k=1}^{\infty} \min(n, k) r^{2k} f_k f_{-k}=\sum_{k=1}^{\infty} \min(k,n) f_k f_{-k},\quad \text{and}$$
$$C_m^{(n)}(f)=\mathcal{O}(R_{in}^{-m})=\mathcal{O}(1),\quad\text{as}\quad n\rightarrow\infty, \text{ for } m\geq 3.$$
In fact as in the proof of Theorem \ref{thm:determinant_expansion} we obtain an explicit bound; once again by Corollary \ref{cor:k-commutators} the sum of the moduli of the coefficients in front of $m$-nested commutators in the expansion of $Z_r^{(m)}(t)$ is at most $\frac{(4e)^m} 
{\sqrt{2\pi m^3}}$. 
Hence by Theorem \ref{thm:bound_on_difficult_nested}, the operators in the expansion of $Z(t)$ (recall \eqref{eqn:z(t)ksum}) satisfy
\begin{enumerate}
\item 
$
\|\mathbf{K}_{\ell}^{(sum)}(f_r)\|_{op}\leq 
\frac{1}{2}\frac{(8Ce)^{\ell}(\ell-1)!}{\sqrt{2\pi \ell^3}},
$
\item $\|P_n\mathbf{K}_{\ell}^{(sum)}(f_r)Q_n\|_{2}, \|Q_n\mathbf{K}_{\ell}^{(sum)}(f_r)P_n\|_{2}
\leq  
\frac{3}{2}
\frac{\kappa^*}{C} \frac{(8Ce)^{\ell}\ell!}{\sqrt{2\pi \ell^3}} \sqrt{\log(1+\ell)}
$,
\item 
$
\left|
\Tr  P_n \mathbf{K}_{\ell}^{(sum)}(f_r) P_n 
\right| 
\leq 
\frac{3}{2}
\left(\frac{\kappa^*}{C}\right)^{2} 
\frac{(8Ce)^{\ell} \ell!}{\sqrt{2\pi \ell^3}}
\frac{\ell}{\ell-1}\log(1+\ell),
\quad\text{for}\quad \ell\geq 3.
$
\end{enumerate}
The additional Stirling approximation, $\frac{\ell!e^{\ell}}{\sqrt{2\pi \ell}}< \ell^\ell e ^{\frac{1}{12\ell}}$ gives us that for $t\leq \frac{1}{24C m}$,
$$\|Z_r^{(m)}(t)\|_{op}
< \frac{1}{12},\quad \|P_nZ_r^{(m)}(t)Q_n\|_2, \|Q_nZ_r^{(m)}(t)P_n\|_2< \frac{1}{8}\frac{\kappa^*}{C},$$
$$\text{and}\quad\left|\Tr  P_nZ_r^{(m)}(t)P_n
- t^2\Tr  P_n\mathbf{K}_2^{(sum)}P_n\right|
\leq
\frac{7}{24}\left(\frac{\kappa^*}{C}\right)^2,
$$
and so as in the proof of Theorem \ref{thm:determinant_expansion}
we see that 
$$|C_m^{(n)}(f_r)|\leq \left(\frac{\kappa^*}{C}\right)^2 \left(\frac{1}{32}A +\frac{7}{24}
\right)\left(24Cm\right)^{m},$$
where $A$ is as in Theorem \ref{thm:generalthm}. Letting $r\uparrow 1$ we conclude the result.
\end{proof}

\appendix

\section{Coefficients in the BCH expansion} \label{appendix:bch}

We briefly discuss how to obtain the concrete Baker-Campbell-Hausdorff expansion in Corollary \ref{cor:k-commutators} from the abstract Theorem \ref{thm:k-commutators}. 
The tool that connects the two is the Dynkin--Specht--Wever theorem, a very useful characterization of
elements in the linear span of $n$-nested commutators (\textit{Lie elements} of order $n$), already used by Dynkin in his proof of Theorem \ref{thm:bch}, \cite{Dyn}.
In fact we believe the discussion to follow to be standard, but for lack of an explicit reference and for completeness, we include it.

We recall the notation introduced in \eqref{eqn:word_nesting}:  a word, $\mathbf{w}\in\mathbb{C}[X_1,X_2,\ldots, X_k]$ of length $n$ is a non-commutative monomial in $X_1, X_2,\ldots, X_k$ of degree $n$, written as
$\mathbf{w}=w_1w_2\ldots w_n$, where $w_j\in\{X_1, X_2,\ldots, X_k\}$. 
We define the linear transformation on the space spanned by words of fixed length $n$ by
\begin{equation}\label{eqn:linear_map_words}
\sigma_n(\mathbf{w})=[\mathbf{w}].
\end{equation}
\begin{theorem}[Dynkin-Specht-Wever]\label{thm:DSW}
An element $\mathbf{a}$ of the linear space spanned by words of length (exactly) $n$ is a Lie element if and only if 
$\sigma_n(\mathbf{a})=n \mathbf{a}.$
\end{theorem}
\begin{proof}
A proof is included in the historical works, \cite{Dyn, Spe, Wev}, and in \cite{BonFul12}. 
\end{proof}
Thus if we can write an expansion for $Z(t)$ in terms of words/monomials we can also write an expansion for $Z(t)$ in terms of nested commutators. The next proposition asserts that there is indeed such an expansion:
\begin{proposition}\label{prop:k_expansion}
Let $Z(t)$ be as in Theorem \ref{thm:k-commutators}, then
\begin{multline}\label{eqn:expansion_k}
Z(t)
=\sum_{m=1}^{\infty}t^m \times
\sum_{j=1}^{m}\frac{(-1)^{j+1}}{j}
\sum_{
\substack{
n_1+n_2+\ldots +n_j=m
\\n_1,n_2,\ldots,n_j\geq 1
}}
\\
\sum_{
\substack{
m_1^{(s)}+m_2^{(s)}+\ldots +m_k^{(s)}=n_s
\\
s=1,2,\ldots, j 
}
}
\frac{X_1^{m_1^{(1)}}\ldots
X_k^{m_k^{(1)}}
\ldots 
\ldots
X_1^{m_1^{(j)}} 
\ldots 
X_k^{m_k^{(j)}}}{m_1^{(1)}! \ldots m_k^{(1)}!
\ldots \ldots
m_1^{(j)}! \ldots m_k^{(j)}!}.
\end{multline}
Moreover, we have that
\begin{equation}
\sum_{j=1}^{m}\frac{1}{j}
\sum_{
\substack{
n_1+\ldots +n_j=m
\\n_1,\ldots,n_j\geq 1
}}
\sum_{
\substack{
m_1^{(s)}+
\\\ldots +m_k^{(s)}=n_s
\\
s=1,2,\ldots, j 
}
}
\frac{1}{m_1^{(1)}! \ldots m_k^{(1)}!\ldots \ldots m_1^{(j)}! \ldots m_k^{(j)}!}\leq \frac{(ke)^m}{\sqrt{2\pi m}}.
\end{equation}
Thus, if we assume that $X_1,\ldots, X_k$ are operators with norm bounded by some $C>0$,
the series expansion, \eqref{eqn:expansion_k}, is valid for $|t|\leq \frac{1}{C2ke}$.
\end{proposition}
We postpone computing this expansion and show how to obtain Corollary \ref{cor:k-commutators} using Theorem \ref{thm:DSW}.
\begin{proof}[Proof of Corollary \ref{cor:k-commutators}]
By Theorem \ref{thm:k-commutators}, formally 
$$Z(t)=\sum_{m=1}^{\infty}t^m\mathbf{K}_m^{(sum)}(X_1,X_2,\ldots, X_k),$$
 for $m\geq 1$, $\mathbf{K}_m^{(sum)}$ is a Lie element of order $m$. %
Proposition \ref{prop:k_expansion} gives a non-commutator form for $\mathbf{K}_m^{(sum)}$, $\mathbf{W}_m$, given by the $m$-th Taylor coefficient in \eqref{eqn:expansion_k}. Then applying $\sigma_m$ to this expression, by Theorem \ref{thm:DSW}, we see that 
\begin{equation}
\frac{1}{m}\sigma_m (\mathbf{W}_m) = \mathbf{W}_m = \mathbf{K}_m^{(sum)},
\end{equation}
and the left-hand side gives the desired expression. 
\end{proof}
We now prove Proposition \ref{prop:k_expansion}.
\begin{proof}[Proof of Proposition \ref{prop:k_expansion}]
For readability we perform the expansion for $k=4$, but the case for any $k\geq 2$ is proved completely analogously.
\begin{multline}
\log(e^{tX_1}e^{tX_2}e^{tX_3}e^{tX_4})
=\log(1+(e^{tX_1}e^{tX_2}e^{tX_3}e^{tX_4}-1))
%
\\=\sum_{j=1}^{\infty}\frac{(-1)^{j+1}}{j} \left(\sum_{n=1}^{\infty} t^n \sum_{d_1+d_2+d_3+d_4=n} \frac{X_1^{d_1} X_2^{d_2}X_3^{d_3} X_4^{d_4}}{d_1!d_2!d_3!d_4!} \right)^j
%
%
\\ =\sum_{j=1}^{\infty}\frac{(-1)^{j+1}}{j} 
\sum_{\substack{
n_1,n_2,\ldots, n_j\geq 1
\\d_1^{(s)}+\ldots+d_4^{(s)}=n_s
\\s=1,2,\ldots,j
}}
t^{\sum_{s=1}^{j}n_s} 
\frac
{X_1^{d_1^{(1)}}\ldots X_4^{d_4^{(1)}}\ldots\ldots X_1^{d_1^{(j)}}\ldots X_4^{d_4^{(j)}}}
{d_1^{(1)}! \ldots d_4^{(1)}!\ldots\ldots d_1^{(j)}! \ldots d_4^{(j)}!}
%
%
\\=\sum_{j=1}^{\infty}\frac{(-1)^{j+1}}{j} 
\sum_{m=j}^{\infty} t^m
\sum_{
\begin{subarray}{c}
 n_1+\ldots
 \\+n_j
=m
\\
n_1,\ldots, n_j\geq 1
\end{subarray}
}
\sum_{\substack{
\\d_1^{(s)}+\ldots
\\+d_4^{(s)}
=n_s
\\s=1,2,\ldots,j
}}
\frac
{X_1^{d_1^{(1)}}\ldots X_4^{d_4^{(1)}}\ldots\ldots X_1^{d_1^{(j)}}\ldots X_4^{d_4^{(j)}}}
{d_1^{(1)}! \ldots d_4^{(1)}!\ldots\ldots d_1^{(j)}! \ldots d_4^{(j)}!}
\\
%
%
=\sum_{m=1}^{\infty} t^m
\sum_{j=1}^{m}\frac{(-1)^{j+1}}{j} 
\sum_{
\substack{
n_1+
\ldots
\\+n_j=m
\\ 
n_1,\ldots, n_j\geq 1
}}
\sum_{\substack{
\\d_1^{(s)}+\ldots
\\+d_4^{(s)}
=n_s
\\s=1,2,\ldots,j
}}
\frac
{X_1^{d_1^{(1)}}\ldots X_4^{d_4^{(1)}}\ldots\ldots X_1^{d_1^{(j)}}\ldots X_4^{d_4^{(j)}}}
{d_1^{(1)}! \ldots d_4^{(1)}!\ldots\ldots d_1^{(j)}! \ldots d_4^{(j)}!},
\end{multline}
where in the last line we exchanged the order of summation. We show that given the assumed bound on the norms, the series converges absolutely (for $t$ sufficiently small), and so the above computation is also valid. We estimate the sum of the absolute value of the coefficients of the terms indexed by $t^m$:
\begin{multline}\label{eqn:coef_estimate}
\sum_{j=1}^{m}\frac{1}{j}
\sum_{
\substack{
n_1+\ldots +n_j=m
\\n_1,\ldots,n_j\geq 1
}}
\sum_{
\substack{
d_1^{(s)}+\ldots +d_4^{(s)}=n_s
\\
s=1,2,\ldots, j 
}
}
\frac{1}{d_1^{(1)}! \ldots d_4^{(1)}!\ldots \ldots d_1^{(j)}! \ldots d_4^{(j)}!}
%
%
\\\leq
\sum_{j=1}^m
\frac{1}{m!}
\sum_{
\substack{
n_1+\ldots +n_j=m
\\n_1,n_2,\ldots,n_j\geq 1
}}
\frac{m!}{n_1!\ldots n_j!}
\sum_{
\substack{
d_1^{(s)}+\ldots d_4^{(s)}=n_s
\\
s=1,2,\ldots, j 
}
}
\frac{n_1!\ldots n_j!}{d_1^{(1)}! \ldots d_4^{(1)}!\ldots\ldots d_1^{(j)}! \ldots d_4^{(j)}!}
%
%
%
\\=
\sum_{j=1}^m
\frac{1}{m!}
\sum_{
\substack{
n_1+\ldots +n_j=m
\\n_1,\ldots,n_j\geq 1
}}
\frac{m! 4^{n_1}}{n_1!\ldots n_j!}
\sum_{
\substack{
d_1^{(s)}
+
\ldots
\\
+ d_4^{(s)}=n_s
\\
s=2,\ldots, j 
}
}
\frac{n_2!\ldots n_j!}{d_1^{(2)}! \ldots d_4^{(2)}!\ldots\ldots d_1^{(j)}! \ldots d_4^{(j)}!}
\end{multline}
\begin{multline*}
=
\sum_{j=1}^m
\frac{1}{m!}
\sum_{
\substack{
n_1+\ldots +n_j=m
\\n_1,\ldots,n_j\geq 1
}}
\frac{m! 4^{\sum_{s=1}^{j}n_s}}{n_1!\ldots n_j!}
=
\sum_{j=1}^m
\frac{4^m}{m!} 
\sum_{
\substack{
n_1+\ldots +n_j=m
\\n_1,\ldots,n_j\geq 1
}}
\frac{m!}{n_1!\ldots n_j!}
\end{multline*}
where in the last two lines we used the multinomial theorem;
$$(x_1+x_2+x_3+x_4)^n= \sum_{d_1+\ldots+d_4=n} \frac{n!}{d_1!d_2!d_3!d_4!} x_1^{d_1}\ldots x_4^{d_4},$$
with $x_1=x_2=x_3=x_4=1.$ Proceeding on we have that
\begin{equation}
LHS\leq \frac{4^m}{m!}
\sum_{j=1}^m
\sum_{
\substack{
n_1+\ldots +n_j=m
\\n_1,\ldots,n_j\geq 1
}}
\frac{m!}{n_1!n_2!\ldots n_j!}
\leq \frac{4^m m^m}{m!}
\leq \frac{(4e)^m}{\sqrt{2\pi m}}.
\end{equation}
In fact there is a shortcut that gives a sharper bound. See that the sum of the absolute value of the coefficients before the terms indexed by $t^m$ on the left-hand side of \eqref{eqn:coef_estimate} corresponds exactly to the $m$-th Taylor coefficient of
$-\log(2-e^{4x})$ (when $x$ is sufficiently small). More generally, for $k\geq 1$, this sum corresponds to the $m$-th Taylor coefficient of $-\log(2-e^{kx})$. But
\begin{multline}
-\log(2-e^{kx})=-\log 2 -\log\left(1-\frac{e^{kx}}{2}\right)
= -\log 2 +\sum_{m=1}^{\infty}\frac{e^{mkx}}{m2^m}
\\=-\log 2 +\sum_{m=1}^{\infty}\frac{1}{m2^m}
\sum_{n=0}^{\infty}\frac{(mkx)^n}{n!}
=
\sum_{n=1}^{\infty}\frac{x^n k^n}{n!}
\sum_{m=1}^{\infty}\frac{m^{n-1}}{2^m}
\\=\sum_{n=1}^{\infty}x^n\frac{k^n}{n!}Li_{-(n-1)}\left(\frac{1}{2}\right),
\end{multline}
and 
$$\frac{k^n}{n!}Li_{-(n-1)}\left(\frac{1}{2}\right)\leq \frac{2k^n}{n(\log 2)^n},\quad \text{for } n\geq 1.
$$
\end{proof}

\end{document}